\documentclass[10pt]{elsarticle}
\RequirePackage{makeidx, amsmath, amsfonts, amssymb, amsthm, setspace, caption, amsthm, subcaption, graphicx, bm}
\usepackage[total={6in, 8in}]{geometry}
\usepackage[utf8]{inputenc}
\usepackage{amsfonts}
\usepackage{amssymb}
\usepackage{amsthm}
\usepackage{amsmath}
\usepackage{mathrsfs}
\usepackage{graphicx}
\usepackage{esvect}
\usepackage{tikz}
\usepackage{algpseudocode}
\usepackage{algorithm}
\usepackage{verbatim}
\usepackage{bbold}
\usepackage{caption}
\usepackage{acronym}
\usepackage{overpic}
\usepackage{amsmath}
\usepackage{xcolor}
\usepackage{sectsty}
\usepackage[section]{placeins}
\usepackage{array, multirow}
\usepackage[normalem]{ulem}
\usepackage{enumitem}
\usepackage{url, hyperref, cleveref}
\usepackage{comment}
\usepackage{booktabs}
\usepackage{longtable}
\usepackage[table]{xcolor}
\usepackage{lscape}
\usepackage{mathrsfs}
\usepackage{mathrsfs}
\usepackage{multicol}
\usepackage{multirow}
\usepackage{rotating}
\usepackage{soul}
\usepackage{placeins}

\usepackage{caption}
\usepackage[normalem]{ulem}

\newtheorem{theorem}{Theorem}[section]

\newtheorem{proposition}{Proposition}[section]

\newtheorem{lemma}{Lemma}[section]

\begin{document}


	\begin{frontmatter}
		\date{\today}
		
        \title{\LARGE Structural and Temporal Hallmarks of Genealogical Networks}
		\author[Carlson]{Japheth Carlson}
        \author[Kim]{Teayoun Kim}
		\author[Lawyer]{Matthew Lawyer}
        \author[Pochman]{Wyatt Pochman}
        \author[Thygerson]{Emeline Thygerson}
		\author[Webb]{Benjamin Webb\footnote{Author to whom any correspondence should be addressed is Benjamin Webb using bwebb@mathematics.byu.edu.}}
		\address[Carlson]{Department of Mathematics, Brigham Young University, Provo, UT 84602, USA, jrc117@student.byu.edu}
        \address[Kim]{Department of Mathematics, Brigham Young University, Provo, UT 84602, USA, rkawk1@student.byu.edu}
		\address[Lawyer]{Department of Mathematics, Brigham Young University, Provo, UT 84602, USA, matthew.lawyer@mathematics.byu.edu}
		\address[Pochman]{Department of Mathematics, Brigham Young University, Provo, UT 84602, USA, wp255@student.byu.edu}
        \address[Thygerson]{Department of Mathematics, Brigham Young University, Provo, UT 84602, USA, emeline.thygerson@byu.edu}
        \address[Webb]{Department of Mathematics, Brigham Young University, Provo, UT 84602, USA, bwebb@mathematics.byu.edu}
		
		\newcommand{\R}{\mathbb{R}}
        
		\begin{abstract}
        The rapid growth of the genealogical sector, spanning platforms with billions of records and millions of users, has produced some of the largest and most complex networks available for analysis. Despite substantial advances in genealogical network research, it remains unclear whether human kinship networks exhibit universal structural properties. We address this by developing an integrated approach to genealogical network analysis that combines network-theoretic structure with an inferred notion of time. Using over one hundred datasets from the Kinsources repository, we reinterpret standard network measures in genealogical terms and introduce \emph{pseudogenerations}, a method for extracting temporal structure directly from network topology.

        Within this framework, we identify common features shared across datasets. We find that genealogical networks exhibit scale-free--like degree and component-size distributions, multiscale family organization, and small-world behavior with respect to genetic and union-based distances. We show that 2-components provide a natural unit of genealogical structure, observe consistent disassortative mixing, and find that recorded unions are strongly biased toward short genetic distances relative to potential pairings. We also document temporal and demographic patterns, including shifts in recorded parental and child information, as well as correlations among recorded unions, parents, and children. These results suggest that diverse genealogical datasets share a common set of structural and temporal characteristics, providing evidence for universal features of human kinship networks and establishing a general framework for their comparative analysis.
		\end{abstract}
		
		
	\end{frontmatter}
	
 
\section{Introduction}\label{sec:intro}

Genealogical networks, which link parents, children, and partners across generations, are of rapidly growing interest due to their relevance in genetics, sociology, demography, and economics \cite{greenwood2014marry, rohde2004common, white2011kinship, malmi2018trends}. Simultaneously, genealogical networks are attracting attention due to individual and public engagement in genealogical research and the growth of large-scale, crowd-sourced genealogy platforms with tens of millions of users \cite{kaplanis2018population, shan2025genealogy, wagner2012attitudes}. This growth is currently driving the development of large-scale online platforms such as Ancestry and FamilySearch, which integrate census, vital, and immigration records into searchable digital systems \cite{ancestryCompanyFacts, familysearch2026facts}. The sector comprising the genealogy products and services market is also projected to grow substantially, with one industry report estimating a~10.1\% annual compound growth rate and a market size exceeding \$8.4 billion by 2030 \cite{businessresearch2026genealogy}.

Together, these developments have produced some of the largest and most richly annotated networks available for scientific analysis. For instance, the FamilySearch dataset, the largest dataset yet compiled, has constructed a human family tree with over 1.86 billion individuals, based on 4.03 billion sources and more than 16.93 billion searchable names in historical records \cite{familysearch2026facts}. The availability of genealogical data at this scale has enabled advances at the intersection of network science, computational social science, population genetics, and topological data analysis. Recent advances include reconstruction of genome-wide genealogies for thousands of individuals, enabling the study of long-range relatedness and demographic history at population scale \cite{kaplanis2018population, speidel2019method}, as well as analyses of spatial dispersal and population structure \cite{osmond2024dispersal}.

Current research in the area uses network-theoretic approaches to characterize genealogical structure in terms of cycles, relinking, and higher-order connectivity \cite{puck2, relinking, white2011kinship, structural-endogamy}. Topological methods, including persistent homology, show that genealogical networks exhibit structural signatures distinct from other social systems, even under subsampling \cite{Boyd2023}. Large-scale demographic studies further examine fertility, mortality, migration, and family structure across populations \cite{kaplanis2018population, nelson2025matrilineal}, while other methodological advances include constructing probabilistic models of genealogical processes \cite{king2022markov} and performing automated record linkage across historical data sources \cite{malmi2018trends, abramitzky2021linking}.

Despite these advances in genealogical network analysis, it remains unclear whether the patterns observed in individual datasets reflect only population-specific characteristics or more universal principles of human kinship organization. Although human societies differ widely in culture, geography, and historical period, kinship networks arise from common biological and social constraints governing reproduction, partnership formation, and ancestry, suggesting the possibility of shared features across populations. Existing studies have primarily focused on individual genealogical datasets \cite{case-study-2,case-study-1}, leaving open the question of whether such commonalities exist and how they evolve over time.

Motivated by this gap, we develop a unified framework for the comparative analysis of genealogical networks and their temporal structure. Using this framework, we investigate whether structural and temporal hallmarks are consistently observed across diverse human kinship networks and, if so, whether these represent universal features of genealogical organization. We apply this framework to over one hundred genealogical networks from \textit{Kinsources.org}, an open-access, peer-reviewed repository of kinship data curated for scientific research \cite{kinsources}. These datasets span a wide range of sizes, ethnicities, cultural contexts, recording practices, and historical periods (see the Appendix for a complete list of datasets), providing a diverse basis for identifying structural and temporal patterns that persist across human populations.

For our structural analysis of these datasets, we recast several of the standard network-theoretic concepts in genealogical terms, including degree, components, distances, assortativity, and clustering coefficients (see Section \ref{sec:spat}). In the context of genealogical networks, \textit{degree} describes an individual's number of parents, children, and partners. \textit{1-components} capture individuals connected via genealogical data while \textit{2-components} identify more cohesive family substructures. We consider three notions of distance: the \emph{standard} graph distance, defined as the length of the shortest path between two individuals; the \emph{genetic distance}, defined as the distance to their nearest common ancestor; and the \emph{distance to union}, which restricts genetic distance to pairs of individuals who form a union. Assortativity, specifically \textit{degree assortativity}, describes the tendency for individuals of high (low) degree to be connected to individuals of high (low) degree, effectively telling us whether individuals from large (or small) families themselves tend to form large (or small) families. \textit{Clustering coefficients} act as a proxy for both the prevalence of unions and the degree to which genealogical networks are locally tree-like.

We find that the distributions of recorded children and unions in the Kinsources datasets are approximately scale-free, suggesting that genealogical datasets are not characterized by a typical family size or number of unions. Similar results hold for the size distributions of 1- and 2-components, which suggests that family information is lost at all scales in genealogical data and that multiscale family structures are a fundamental property of these networks (see Section \ref{sec:spat}). In this context, we prove under mild conditions that $k$-components for $k\geq3$ fracture genealogically meaningful structures, separating children from parents (see the Appendix). These genealogical networks also exhibit, on average, disassortative mixing, which becomes stronger as unions are removed from the network. This indicates that children from large (small) families tend to have smaller (larger) families, a pattern that may reflect either genealogical structure itself or patterns of recording genealogical data. Finally, because of the prevalence of triangles induced by unions in local family subgroups, the clustering coefficients of genealogical networks are quite high relative to other real-world networks (see Section \ref{sec:spat}).   

In addition to these network properties, we introduce and examine genealogically specific measures related to individuals, such as parent type statistics, average numbers of children by parent type, and parent-child-union correlations. These quantities describe the fractions of different parent types (which depend on completeness of familial data), the average number of children associated with each parental type, and the correlations between each individual's number of recorded parents, children, and unions (see Section \ref{sec:spat}).

Beyond these structural features we also consider the temporal evolution of genealogical networks. As previously mentioned, a major feature that distinguishes genealogical networks from other real-world networks is an implicit temporal structure. That is, parent-child edges are always directed forward in time. Furthermore, individuals and relationships, once properly recorded, remain in the genealogical network. These two features allow us to place a temporal structure on a genealogical network giving it a generational structure, without reference to anything other than the network's graph structure. This, in turn, allows for a temporal analysis of the network despite the fact that most genealogical datasets are static, not dynamic.

To study this temporal structure, we introduce the notion of pseudogenerations, which provide an approximate temporal ordering of individuals in a genealogical network (see Section \ref{sec:pseudo}). This construction, phrased as a network optimization problem (see Equation~\eqref{eq:pop}), enables the study of the evolution of structural properties across generations, yielding an integrated structural and temporal framework for characterizing the defining features of genealogical networks.

Computing pseudogenerational structures for the Kinsources datasets enables a temporal analysis of the structural features introduced in Section \ref{sec:spat}. We find that genetic distances and union-based distances exhibit approximately logarithmic growth across pseudogenerations, suggesting that genealogical networks display small-world behavior with respect to these metrics (see Figure \ref{fig:9} (left)). In particular, while individual families form tightly clustered substructures, individuals in these families remain connected globally through a relatively small number of genealogical links. Using the union-based distance, we further find that recorded unions are strongly biased toward short genetic separations, with half of all unions occurring within the closest 15.63\% of potential pairings (see Figure \ref{fig:9} (right)).

For the parent type statistics, we observe a pronounced temporal shift: there is a strong decrease in the fraction of individuals with two recorded parents, while the fraction of parents with no recorded children has a nearly monotonic increase over time (see Figure \ref{fig:11} (left)). For the correlation structure, we find a strong positive association between an individual's number of recorded children and number of recorded unions, which increases over time. Perhaps surprisingly, there is effectively no correlation between an individual's number of recorded parents and number of recorded children and a weak but negative correlation between an individual's number of recorded parents and number of recorded unions (see Table  \ref{tab:2} and Figure \ref{fig:11} (right)).    

In summary we make the following contributions: 

\begin{itemize}
\item We introduce a unified framework for the comparative analysis of genealogical networks that integrates structural and temporal perspectives through the notion of pseudogenerations.

\item We adapt and reinterpret core network-theoretic concepts in genealogical terms, providing a methodology for studying family-based networks using only network structure.

\item Across more than one hundred genealogical datasets, we identify common structural hallmarks of human kinship networks, including scale-free--like degree and component-size distributions, multiscale family organization, and small-world behavior with respect to genetic and union-based distances.

\item We identify 2-components as a natural unit of genealogical structure, balancing connectivity and interpretability, and show that higher-order connectivity fragments meaningful genealogical relationships.

\item We identify common interaction patterns across datasets, including consistently disassortative degree correlations and the preferential formation of unions at relatively short genetic distances.

\item We identify common temporal and demographic patterns, including shifts in recorded parental and child information, as well as correlations among recorded unions, parents, and children.

\end{itemize}

The paper is organized as follows. In Section~\ref{sec:background}, we describe the Kinsources datasets analyzed in this work and our framework for representing genealogical networks. Section~\ref{sec:spat} examines the structural (topological) features of these data, which we term structural hallmarks. In Section~\ref{sec:data}, we introduce pseudogenerations, a method for inferring temporal structure from genealogical data, and use it to study the temporal evolution of these features, which yield temporal hallmarks. We conclude by discussing directions for future work enabled by these results. Algorithms, proofs of analytic results, and a table of the Kinsources datasets are provided in the Appendix.

\section{Background and Data}\label{sec:background}

Unless otherwise noted, the datasets analyzed in this study comprise 105 genealogical networks drawn from \textit{Kinsources.org}, an open-access, peer-reviewed repository of kinship data curated for scientific research \cite{kinsources}. Of the 123 total Kinsources datasets, we excluded 18. Of the 18 excluded networks, seventeen are smaller duplicates of networks already included and one network (\textit{Udalen}) contains neither individuals nor edges.

We represent each dataset as a graph $G_i$ for $i=1,\dots,105$, whose nodes correspond to individuals and whose edges encode genealogical relationships. The resulting networks vary substantially in size, edge density, connectivity structure, and completeness, as well as in ethnicity, cultural context, recording practices, and historical period. The largest network, the \textit{Watchi} network, contains approximately 50{ }000 individuals and 83{ }000 edges, whereas the smallest, the ``\textit{Family}'' network, consists of 17 individuals and 24 edges. The complete list of the datasets $\{G_i\}_{i=1}^{105}$ is provided in the Appendix.

We also consider the aggregate network $A_G=\bigsqcup_{i=1}^{105}G_i$, the disjoint union of the 105 Kinsources datasets. We treat individuals in this dataset as distinct; that is, we do not attempt entity resolution across datasets. This convention allows us to characterize statistical properties at the level of individuals in the aggregate network, whereas analysis of the collection $\{G_i\}_{i=1}^{105}$ enables us to study statistical properties at the dataset level (cf. Tables \ref{tab:0}--\ref{tab:2}).

To represent these datasets, we use the notion of an Ore graph, a standard graph-theoretic representation of genealogical data in which individuals are modeled as vertices, parent--child relationships as directed arcs, and unions as undirected edges \cite{batagelj2008analysis}. The Ore graph representation models a genealogical network as a mixed graph $G=(N,E)$, where $N=\{1,2,\dots,n\}$ denotes the set of individuals, and the edge set $E=E_U\cup E_C$ consists of directed parent–child edges $E_C$ and undirected union edges $E_U$ (cf. Figure \ref{fig:1}). Directed edges encode the inherent asymmetry of parent–child relationships, while undirected union edges represent marital or spousal-type relationships between \textit{partners}, whose precise interpretation may vary across datasets. For a given network, we let $m$ denote its number of edges. For instance, the aggregate network has $n=249\text{ }597$ individuals and $m=439\text{ }613$ edges.

Depending on the specific type of analysis, we also consider simplified versions of this representation. In some cases, we work with the undirected version of the network, $G=(N,E)$, obtained by ignoring the directionality of parent-child relationships. In other cases, we focus on the \textit{genetic} version of the network, $\mathcal{G}=(N,E_C)$, which is formed by removing all union edges from the genealogical network. For most genealogical networks, the associated genetic network $\mathcal{G}$ is a directed acyclic graph (DAG), reflecting the temporal ordering inherent in genealogical relationships. While recording errors or inconsistencies can, in principle, introduce cycles, all 105 genetic networks analyzed in this study are DAGs.

\begin{figure}
\begin{center}
    \begin{tabular}{ccc}
        \begin{overpic}[scale=0.1]{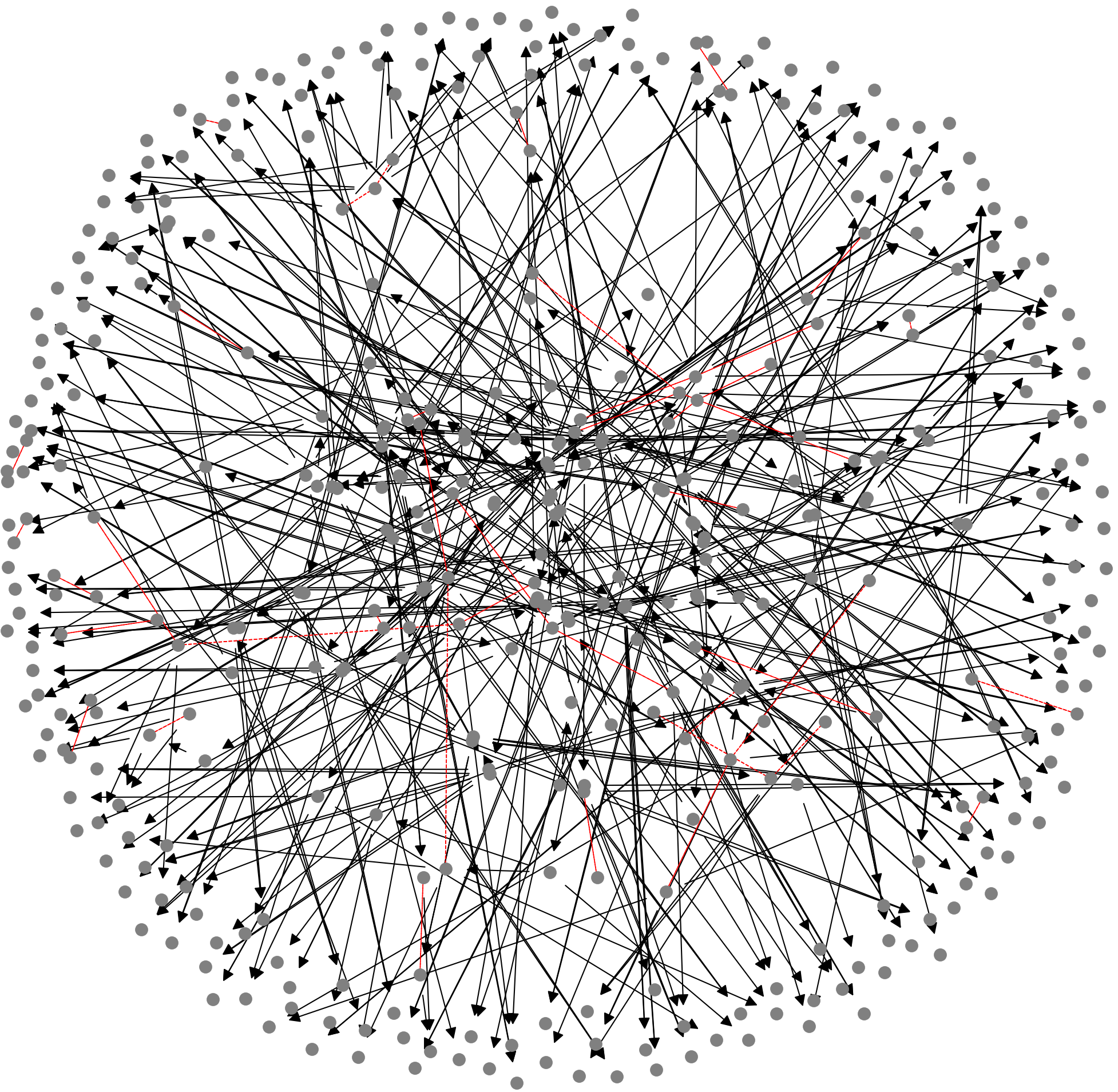}
        \put(-17.5,-7){Dogon--Konsogu--Donyu Genealogical Network}
        \end{overpic}
    \end{tabular}
    \vspace{0.05cm}
\caption{The Dogon--Konsogu--Donyu genealogical network $G=(N,E)$ is shown consisting of $n=399$ individuals and $m=592$ edges. The edge set $E=E_U\cup E_C$ comprises $|E_U|= 73$ undirected union edges, shown in red, and $|E_C|= 519$ directed parent-child edges shown in black (cf. Figure \ref{fig:7}).}
\label{fig:1}
\end{center}
\end{figure}

\section{Structural Hallmarks of Genealogical Networks}\label{sec:spat}
Genealogical networks are often described as “tree-like” because from the perspective of a single individual, ancestry is traced through parents, grandparents, and earlier generations. When attention is restricted to a small number of generations and union edges are ignored, this ancestral subgraph is indeed a tree. However, full genealogical networks are substantially more complex (for instance, see the Dogon--Konsogu--Donyu genealogical network in Figure \ref{fig:1}).

Even at the local level, cycles arise from elementary family structures, such as the triangles formed by two parents and a child. At larger scales, additional cycles emerge through shared ancestry: tracing sufficiently far back in time typically yields a nearest common ancestor of an individual’s parents. The inclusion of union edges further increases this complexity, generating larger and more complicated cycles that span multiple families and generations.  

In this section, we examine the topological features found in these datasets, which we refer to as \textit{structural hallmarks}. In Section \ref{sec:data}, we revisit these same features from a temporal perspective in order to study their evolution over time. The statistics we consider are drawn largely from network science, but are selected for their natural genealogical interpretation. These include standard measures such as degree distributions, $k$-components, assortativity, and clustering coefficients, as well as additional quantities specific to genealogical networks, including the average numbers of children, parents, and unions per individual (cf. Tables \ref{tab:0} and \ref{tab:2}).\\

\begin{figure}
\begin{center}
    \begin{tabular}{ccc}
        \hspace{-0.25cm}
        \begin{overpic}[scale=0.3000]{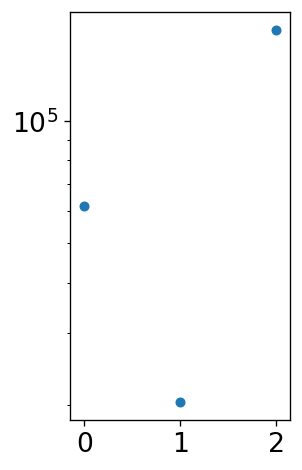}
        \put(10.5, -9){\small{Parent Distribution}}
        \put(15, -1.5){\scriptsize{$\mathrm{Number \ of \ Parents}$}}
        \put(5, 30){\tiny{\rotatebox{90}{$\textrm{Individual Count}$}}}
        \end{overpic} \hspace{0.33cm}
        
        \begin{overpic}[scale=0.3000]{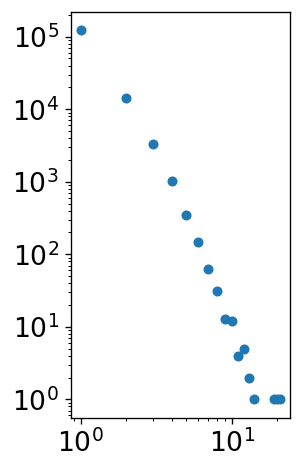}
        \put(10.5, -9){\small{Union Distribution}}
        \put(15, -1.5){\scriptsize{$\mathrm{Number \ of \ Unions}$}}
        \put(-4, 30){\tiny{\rotatebox{90}{$\textrm{Individual Count}$}}}
        \end{overpic} \hspace{0.33cm}
        
        \begin{overpic}[scale=0.3000]{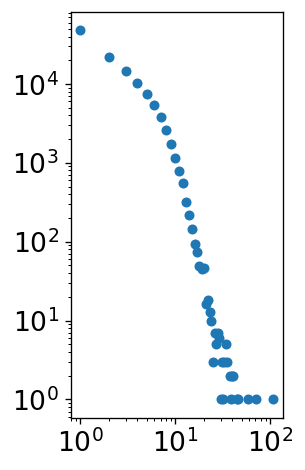}
        \put(13, -1.5){\scriptsize{$\mathrm{Number \ of \ Children}$}}
        \put(-4, 30){\tiny{\rotatebox{90}{$\textrm{Individual Count}$}}}
        \put(12,-9){\small{Child Distribution}}
        \end{overpic} \hspace{0.33cm}
        
         \begin{overpic}[scale=0.3000]{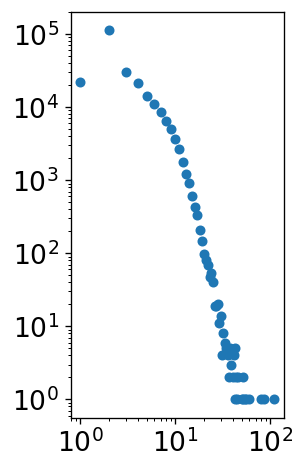}
         \put(21, -1.5){\scriptsize{$\mathrm{Total \ Degree}$}}
        \put(-4, 30){\tiny{\rotatebox{90}{$\textrm{Individual Count}$}}}
        \put(10,-9){\small{Degree Distribution}}
        \end{overpic}
        \vspace{10pt} 
    \end{tabular}
\caption{
The parent, union, child, and total degree distributions for the aggregate Kinsources dataset \(A_G=\bigsqcup_{i=1}^{105}G_i\) are shown from left to right. The parent distribution is shown on a semi-log scale. For the union, child, and total degree distributions, positive degrees are shown on log--log axes. Of these individuals, approximately 43.01\% have no recorded union, 52.00\% have no recorded children, and 1.60\% have a total degree distribution of zero. The union and child positive-degree distributions exhibit approximately scale-free behavior over the observed range, which is reflected in the total degree distribution.  This suggests that the recorded number of unions and children lacks a characteristic scale in the Kinsources dataset.}\label{fig:2}
\end{center}
\end{figure}
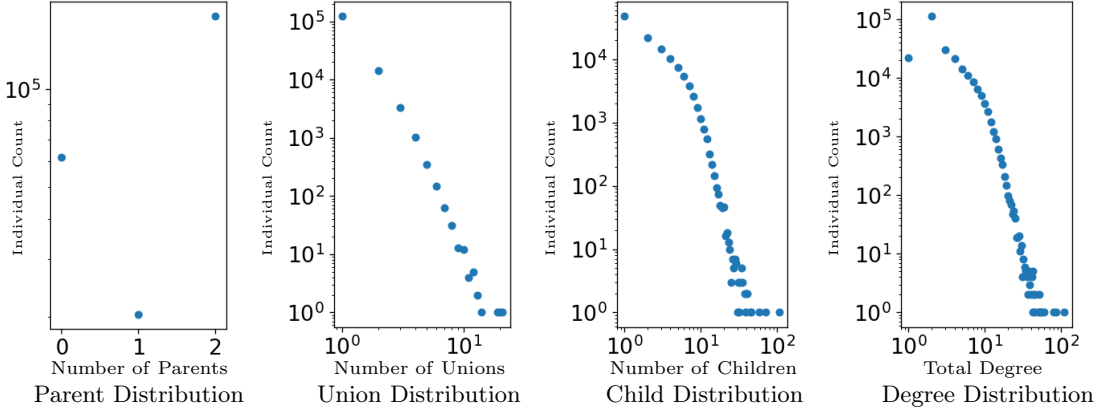

\begin{table}
\centering
\begin{tabular}{c|c|c|c|c|c|c}
\toprule
\toprule
Quantity &$n$&$m$&$p$& u & c & k\\
\midrule
$\langle \cdot \rangle_A$ &$249\text{ }597$& $439\text{ }613$& $1.42$& $0.68$& $1.42$& $3.52$\\
\midrule
$\langle \cdot\rangle_{i}$ &$2377.11$& $4186.79$& $1.31$& $0.75$& $1.31$& $3.37$\\
$\sigma_i(\cdot)$ & $6107.91$& $10\text{ }460.09$& $0.26$& $0.22$& $0.26$& $0.54$\\
\bottomrule
\bottomrule
\end{tabular}
\caption{Summary statistics for size and degree properties of genealogical networks. Average values $\langle\cdot\rangle_A$ are reported for the aggregate dataset $A_G = \bigsqcup_{i=1}^{105}G_i$ (first row), along with the mean $\langle\cdot\rangle_i$ and standard deviation $\sigma_i(\cdot)$ across the 105 individual Kinsources datasets $\{G_i\}_{i=1}^{105}$ (second and third rows, respectively). Columns report, from left to right, the number of individuals $n$ and edges $m$, the average number of parents $p$, number of unions $u$, number of children $c$, and average degree $k$, respectively.} 
\label{tab:0}
\end{table}

\textbf{Degree Distribution:} The first structural feature we consider is degree distribution. For simplicity, we work with the undirected version of a genealogical network.  In this representation, individual $i$ has degree $k^i=k^i_{p}+k^i_{u}+k^i_{c}$, where $k^i_p$, $k^i_u$, and $k^i_c$ denote the individual's number of recorded parents, unions, and children, respectively. The corresponding in- and out-degree distributions for the directed version of the network can be recovered from these distributions as $k^i_{\text{in}}=k^i_{p}+k^i_{u}$ and  $k^i_{\text{out}}=k^i_{u}+k^i_{c}$ for the individual $i$, respectively. The average number of recorded parents $\langle p\rangle$, unions $\langle u\rangle$, children $\langle c\rangle$, and degree $\langle k\rangle$ for both the aggregate dataset $A_G=\bigsqcup_{i=1}^{105}G_i$ and the individual datasets $\{G_i\}_{i=1}^{105}$ are shown in Table \ref{tab:0}. For the individual datasets, the corresponding standard deviations are also listed. Notably, the average numbers of recorded parents, unions, and children show similar standard deviations across the Kinsources datasets, suggesting comparable variability driven by cultural, temporal, or recording differences.

Figure \ref{fig:2} shows the parent, union, child, and total degree distributions, respectively, for the aggregate dataset $A_G=\bigsqcup_{i=1}^{105}G_i$. The parent distribution is shown on a semi-log scale, while the positive-degree union, child, and total degree distributions are shown on log--log axes. Across the data, 67.0\% of individuals have two parents recorded, 8.2\% have one parent recorded, and 24.8\% have no recorded parents. Since each individual necessarily has two biological parents, this indicates a loss of 28.88\% of the parental edges the corresponding complete genealogical dataset would have. This notion of missing information is revised later in this section when we consider parent-union-child statistics. 

The distribution of unions in Figure \ref{fig:2} displays a strikingly linear trend, suggesting a scale-free structure in which the probability of having union degree of $k_u$ is $P(k_u)\sim k_u^{-\alpha_u}$ for $k_u\geq k_u^{min}$ for some cutoff $k_u^{min}$ and \textit{scaling exponent} $\alpha_u$. Using the maximum-likelihood estimators
\begin{align}\label{eq:slope}
\alpha &= 1 + N \left( \sum_i \ln \frac{k^i}{k^{\min}-1/2} \right)^{-1}, \\
\sigma &= (\alpha - 1)/\sqrt{N},\label{eq:error}
\end{align}
on the generic quantity $k$, we estimate the scaling exponent $\alpha_u = 1.8765$, corresponding to a log--log slope of approximately $-\alpha_u$, with an \textit{error} of $\sigma_u = 0.0018$ for the union distribution when $k_u^{\min}=1$. The distribution of children shows a similar scale-free--like behavior, although over a larger range, with $\alpha_c = 2.9345$, $\sigma_c = 0.0087$, and $k_c^{\min}=4$. 

This suggests that there is no typical number of unions or family size to be found in these genealogical datasets. While the precise mechanism producing this pattern remains unclear, similar behavior is observed across many of the Kinsources datasets, suggesting the effect is not solely due to aggregation across multiple cultures and time periods, but instead reflects a more general and potentially surprising feature of genealogical network structure.\\ 

\begin{figure}
   \centering
    \begin{overpic}[scale=0.5]{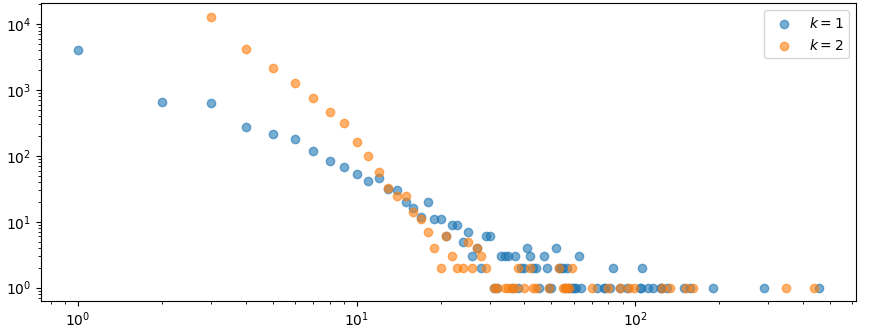}
        \put(16,-5){\scriptsize{1- and 2-Component Size Distributions of the Aggregate Dataset}}
        \put(42,-1){\scriptsize{$\mathrm{Component \ Size}$}}
        \put(-3,9){\tiny{\rotatebox{90}{$\textrm{Number of Components}$}}}
    \end{overpic} 
    \vspace{0.45cm}
   \caption{1- and 2-component size distributions (blue and orange, respectively) for the aggregate dataset $A_G=\bigsqcup_{i=1}^{105}G_i$ are shown on a log--log plot. Both distributions exhibit an approximately linear trend, suggesting scale-free–like behavior in component sizes. For $k=1$, this indicates that the groups of individuals connected by genealogical data occur over multiple scales with no typical size (i.e. the size of such groups do not cluster around a particular size). For $k=2$, a similar pattern is observed for more tightly connected family units, indicating that these higher-order family groups also occur across a wide range of sizes with no typical scale.}
   \label{fig:3}
\end{figure}

\textbf{k-Components:} A \textit{$k$-component} of a network $G=(N,E)$ is a maximal vertex set $C \subseteq N$ such that every pair of individuals in $C$ is connected by at least $k\geq1$ internally node-disjoint paths contained in the subgraph induced by $C$. For this analysis, we treat genealogical networks as undirected graphs and focus on $k=1$ and $k=2$. When $k=1$, $k$-components coincide with the connected components of the network. For $k=2$, the resulting $k$-components capture cohesive familial substructures, whereas higher-order components ($k \ge 3$) tend to fragment these structures into less interpretable units. 

This distinction arises from features that have a natural interpretation in genealogical networks. If two parents belong to a 2-connected component, then, under the assumption that the network is \textit{union-complete} (see the Appendix for the precise definition), each of their children necessarily belongs to the same component. In contrast, for $k\geq3$, every $k$-connected component must contain at least one set of partners whose children lie outside the component. Consequently, decompositions into $k$-connected components with $k\geq3$ may separate parents from their children, potentially obscuring the genealogical meaning of these components. Proofs of these structural properties are provided in the Appendix (see Lemma \ref{lemma:pc_edges} and Proposition \ref{thm:3-connected}).

For $k = 1$, the aggregate network $A_G = \bigsqcup_{i=1}^{105}G_i$ exhibits the component-size distribution shown in Figure~\ref{fig:3} (blue). The distribution follows an approximately linear trend on a log--log scale, consistent with a scale-free--like structure. Using the maximum-likelihood estimators defined in Equations~\eqref{eq:slope} and \eqref{eq:error}, we obtain the scaling exponent $\alpha_1 = 2.375$ with standard error $\sigma_1 = 0.046$, computed using $k^{\min}_1 = 6$. This indicates that the number of individuals connected through genealogical data spans multiple scales and lacks a characteristic component size. Similar behavior has been observed in other genealogical networks, including the much larger FamilySearch dataset, which contains over one billion individuals (see Figure 3 in~\cite{rapp2012analyzing}).

\begin{figure}
    \centering
    \begin{overpic}[scale=0.55]{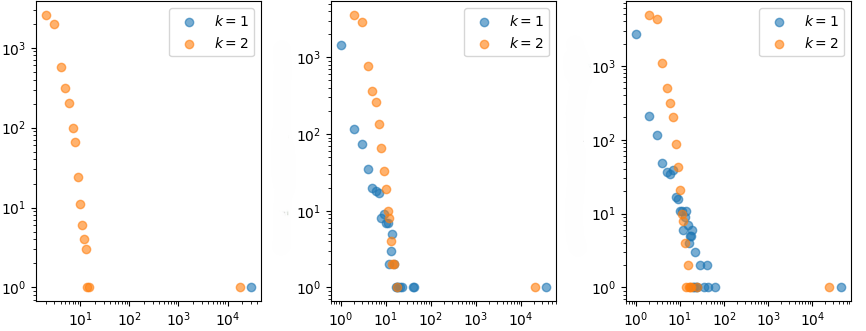}
        \put(11,-4.5){\scriptsize{Watchi 2008}}
        \put(9,-1.5){\scriptsize{$\mathrm{Component \ Size}$}}
        \put(-2.5,10){\tiny{\rotatebox{90}{$\textrm{Number of Components}$}}}
        
        \put(44,-4.5){\scriptsize{Watchi 2011}}
        \put(42.5,-1.5){\scriptsize{$\mathrm{Component \ Size}$}}
        \put(32,10){\tiny{\rotatebox{90}{$\textrm{Number of Components}$}}}
        
        \put(78,-4.5){\scriptsize{Watchi 2017}}
        \put(76.5,-1.5){\scriptsize{$\mathrm{Component \ Size}$}}
        \put(66,10){\tiny{\rotatebox{90}{$\textrm{Number of Components}$}}}
    \end{overpic}
    \vspace{0.5cm}
    \caption{The component size distributions of the Watchi genealogical dataset at three recording snapshots, 2008, 2011, and 2017, are shown from left to right, respectively. Blue and orange curves show the size distributions of 1-components and 2-components, respectively, on a log--log scale. In 2008 (left), a single 1-component contains all 2-components, where the 2-component distribution has a scale-free-like shape. Subsequent additions to the Watchi dataset (center and right) introduce new components and produce approximately scale-free distributions for the network's 1-components, while effectively preserving the distribution of the network's 2-components.}
    \label{fig:watchi_distrs}
\end{figure}
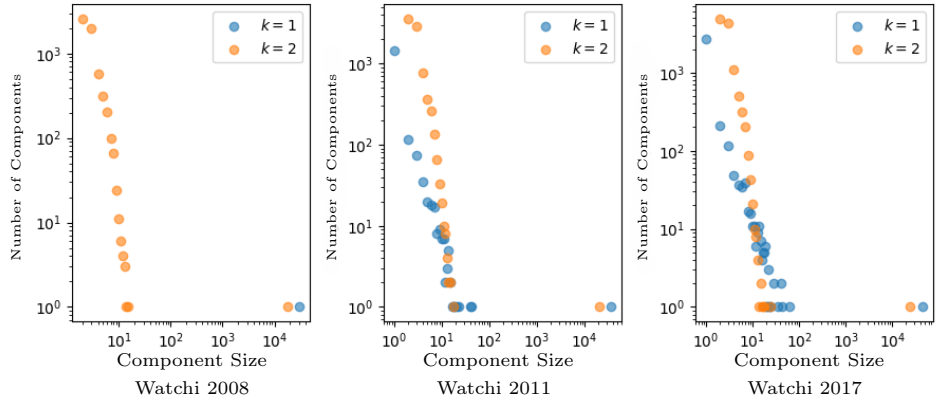

A comparable pattern is observed in the Watchi family, the largest dataset in Kinsources (cf.\ Figure~\ref{fig:watchi_distrs}, right panel). The Watchi data consist of three temporally distinct snapshots (2008, 2011, and 2017), with progressively more genealogical information incorporated over time; only the largest snapshot is used in the aggregate network $A_G$. Figure~\ref{fig:watchi_distrs} presents these snapshots from left to right, showing the distributions of 1- and 2-components (blue and orange, respectively) on log--log axes. As additional data are incorporated, the scale-free--like behavior of the 1-component  distribution becomes increasingly pronounced. That is, the inclusion of more data over time leads to the appearance of additional disconnected components, revealing more gaps in the data. The scale-free--like nature of the 1-components suggests that missing genealogical information (i.e. missing individuals and/or relationships) may be similarly distributed across multiple scales.

Moving to substructures within connected components, individuals within the same family are often connected by multiple independent paths arising from shared ancestry, sibling relations, and unions. These redundant paths create cycles, so that the removal of any single individual does not disconnect the subgraph. For $k=2$, a network's 2-components naturally delineate these more cohesive family units. The 2-component size distribution of the aggregate network is shown in Figure \ref{fig:3} (orange, log--log scale) and also exhibits a strikingly linear trend, with $\alpha_2=3.372$, $\sigma_2=0.016$, and $k_2^{\min}=3$. Unlike the 1-components, the scale-free pattern of 2-components in datasets like Watchi (cf. Figure \ref{fig:watchi_distrs}) appears more robust to the loss of information, though the precise underlying mechanism for this remains unclear.

This suggests that typical genealogical networks possess scale-free--like connectivity at both the level of global connected components (1-components) and local family subunits (2-components). The consistency of these patterns across the individual Kinsources and aggregate datasets suggests that such hierarchical, scale-free--like structure appears to be a fundamental property of genealogical networks.\\

\begin{figure}
\begin{center}
    \begin{tabular}{cc}
        \begin{overpic}[scale=0.22]{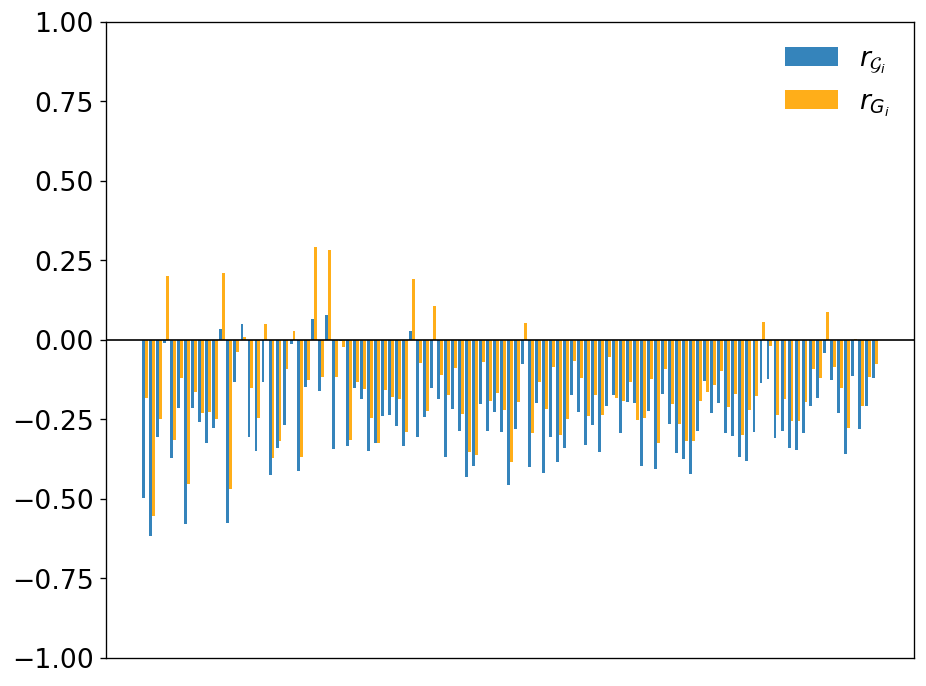}
        \put(36,-8.5){\scriptsize{Degree Assortativity}}
        \put(30,-3.5){\scriptsize{$\mathrm{Individual \ Datasets:\{G_i\}_{i=1}^{105}}$}}
        \put(-4,28){\tiny{\rotatebox{90}{$\textrm{Assortativity}$}}}
        \end{overpic} 
        \hspace{0.15cm}
         \begin{overpic}[scale=0.22]{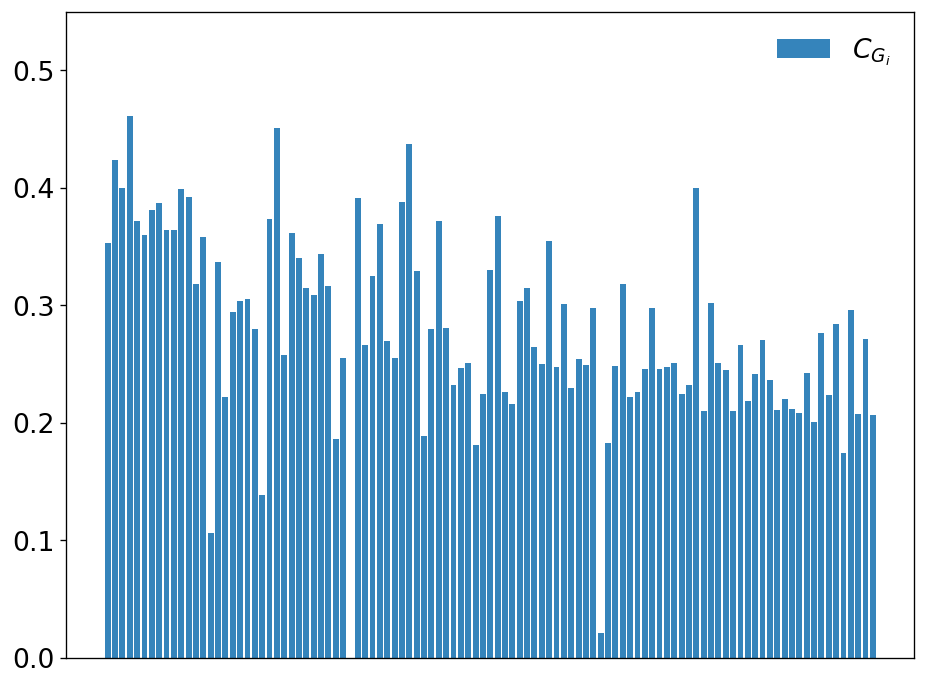}
        \put(26,-8.5){\scriptsize{Global Clustering Coefficient}}
        \put(26,-3.5){\scriptsize{$\mathrm{Individual \ Datasets:\{G_i\}_{i=1}^{105}}$}}
        \put(-4, 22.5){\tiny{\rotatebox{90}{$\textrm{Clustering Coefficient}$}}}
        \end{overpic} 
    \end{tabular}
    \vspace{0.35cm}
\caption{(Left) The degree assortativity coefficients of the genealogical datasets $\{G_i\}_{i=1}^{105}$ and the corresponding genetic datasets $\{\mathcal{G}_i\}_{i=1}^{105}$ are shown in orange and blue, respectively. (Right) The clustering coefficients of the genealogical datasets $\{G_i\}_{i=1}^{105}$ are shown. In both panels, datasets are ordered from left to right by increasing size, with the smallest dataset on the left and the largest on the right.}\label{fig:5}
\end{center}
\end{figure}
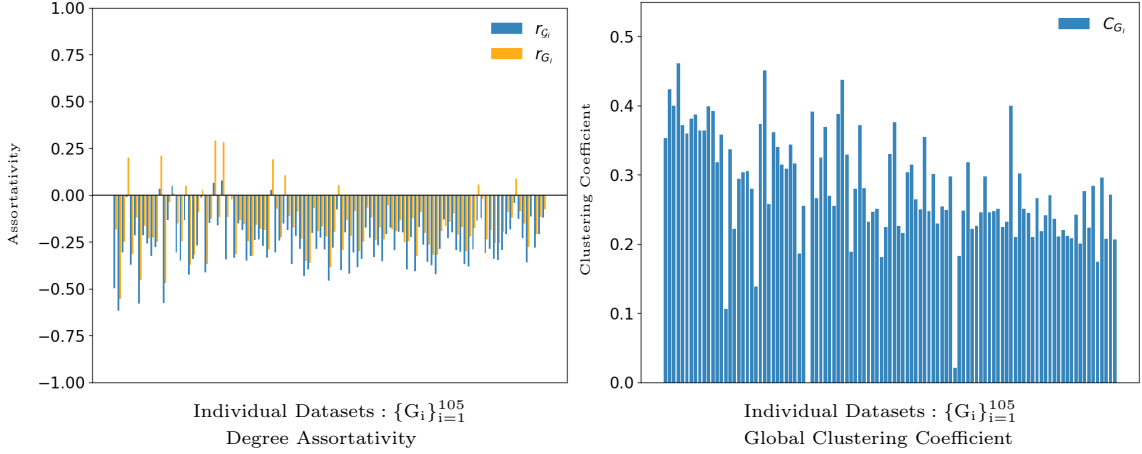

\textbf{Assortativity:} Assortativity measures the tendency of nodes to be connected to other nodes which share similar attributes. In network analysis, degree assortativity, the tendency for high- (or low-) degree nodes to connect preferentially to nodes of similar degree, is of particular interest. In the context of genealogical networks, degree assortativity effectively captures whether individuals from large (or small) families tend to form families of comparable size or form unions with families of comparable size. Here, an individual’s degree is determined by the number of recorded parents, children, and partners in the genealogical record.

For an undirected genealogical network $G$, the \textit{adjacency matrix} $A=[a_{ij}]$ has zero-one entries where $a_{ij}=1$ if individual $i$ shares a parent-child or union edge with individual $j$, and is zero otherwise. The degree assortativity coefficient of the network is given by
\[
r=\frac{\sum_{i,j}(a_{ij}-k^ik^j/2m)k^ik^j}{\sum_{i,j}(\delta_{ij}k^i-k^ik^j/2m)k^ik^j},
\]
where $\delta_{ij}$ denotes the Kronecker delta. The coefficient $r\in[-1,1]$, with positive values indicating assortative mixing and negative values indicating disassortativity.

For the aggregate genealogical network $A_G$, we obtain a degree assortativity of $r_G=-0.10$, indicating a tendency toward disassortativity. For comparison, the genetic version of the aggregate dataset has degree assortativity $r_{\mathcal{G}}=-0.16$, suggesting that the exclusion of union edges increases the disassortativity (see Table \ref{tab:1}). A likely explanation is that individuals connected through unions share family, leading to more comparable degrees; removing these edges therefore accentuates disassortativity. 

\begin{table}[b]
\centering
\begin{tabular}{c|c|c|c|c|c|c}
\toprule
\toprule
Quantity & $r_G$ & $r_{\mathcal{G}}$ & $C_G$ & $\ell_s$ & $\ell_g$ & $\ell_u$ \\
\midrule
$\langle \cdot \rangle_A$ & $-0.10$& $-0.16$& $0.23$& $10.97$ & $11.83$& $7.79$\\
\midrule
$\langle \cdot\rangle_{i}$ & $-0.16$& $-0.26$& $0.28$& $7.59$ & $4.99$ & $5.71$\\
$\sigma_i(\cdot)$ & $0.15$& $0.13$& $0.08$& $3.87$ & $2.50$ & $1.84$\\
\bottomrule
\bottomrule
\end{tabular}
\caption{Summary statistics for macroscopic properties of genealogical networks. Average values $\langle\cdot\rangle_A$ are reported for the aggregate dataset $A_G = \bigsqcup_{i=1}^{105}G_i$ (first row), along with the mean $\langle\cdot\rangle_i$ and standard deviation $\sigma_i(\cdot)$ across the 105 individual Kinsources datasets $\{G_i\}_{i=1}^{105}$ (second and third rows, respectively). Columns report, from left to right, the degree assortativity of the genealogical network $r_G$, the degree assortativity of the associated genetic network $r_\mathcal{G}$, the clustering coefficient $C_G$, the mean standard distance $\ell_s$, the mean genetic distance $\ell_g$, and the mean distance to union $\ell_u$.}
\label{tab:1}
\end{table}

A similar pattern is observed across the individual Kinsources datasets $\{G_i\}_{i=1}^{105}$.  The mean assortativity is $\langle r_G\rangle_i=-0.16$ with standard deviation $\sigma_i(r_G)=0.15$, whereas the corresponding genetic networks have a much larger disassortativity, given by $\langle r_\mathcal{G}\rangle_i=-0.26$ with $\sigma_i(r_\mathcal{G})=0.13$. On average, the genetic networks exhibit substantially larger disassortativity than the full genealogical networks. This behavior is remarkably consistent (see Figure~\ref{fig:5} (left)): 92 of the 105 genealogical networks (87.6\%) and 100 of the 105 corresponding genetic networks (95.2\%) exhibit negative degree assortativity.

Overall, the assortativity values are slightly to moderately negative, indicating that individuals with high recorded degree tend, on average, to be connected to individuals with lower recorded degree. In genealogical terms, this suggests a tendency for the size of the recorded family in which an individual is raised to differ from the size of the recorded family they later form or join through a union. This finding contrasts with \cite{family-size}, which reports a small but statistically significant intergenerational correlation in family size. The difference between these findings may stem from differences between recorded and actual family size, as well as differences in how family size is defined, although the exact difference is currently unknown.\\ 

\textbf{Clustering Coefficients:} The \textit{local} clustering coefficient $C_i$ of a node $i$ measures the probability that any two neighbors of node $i$ are connected, forming a triangle. This is calculated as the ratio of actual triangles to potential triangles for an individual. Given an undirected genealogical network $G$, this can be extended to the global clustering coefficient, which can be expressed as
\[
C=\frac{\text{number of triangles}\times3}{\text{number of paths of length 2}}.
\]

Because triangles arise only in the presence of unions, a low global clustering coefficient indicates that either relatively few parents in the network are connected by a union or individuals are missing one or more parents. In either case, the network becomes more locally tree-like, with individuals participating in fewer short cycles. Thus, the global clustering coefficient serves as a proxy both for the prevalence of unions and for the degree to which the genealogical network locally departs from the structure of a tree.

For the aggregate dataset \(A_G\), we obtain a global clustering coefficient of \(C_G = 0.23\). Across the Kinsources datasets, the mean clustering coefficient is \(\langle C \rangle_i = 0.28\), with standard deviation \(\sigma_i(C) = 0.08\) (see Table \ref{tab:1}). These coefficients are relatively high compared to most real-world networks, reflecting the natural formation of parent--child triangles in genealogical structures. We note that the associated genetic network has a clustering coefficient $C_{\mathcal{G}}$ of zero because removing union edges eliminates the parent-parent-child triangles present in the full genealogical network. The clustering coefficients of the individual datasets $\{G_i\}_{i=1}^{105}$ are shown in Figure \ref{fig:5} (right).\\

\textbf{Distances:} The standard distance matrix $D = [d_{ij}]$ of a genealogical network $G$ encodes the shortest-path distances between all pairs of individuals $i$ and $j$, with $d_{ij} = \infty$ if no path exists. Here, we consider standard distance and two other types of distances:
\begin{itemize}
    \item \textit{Standard Distance:} the shortest-path distance in the undirected version of $G$.
    \item \textit{Genetic Distance:} the shortest-path distance in the undirected genetic network $\mathcal{G}$ through a common ancestor.
    \item \textit{Distance to Union:} the genetic distance restricted to pairs of individuals connected by a union in $\mathcal{G}$.
\end{itemize}
Standard distance captures the mix of genetic separation and separation by union in the full genealogical network. Genetic distance provides a proxy for biological relatedness, yielding a lower bound on genetic similarity between individuals. Distance to union measures the genetic separation between partners, allowing us to quantify the genetic distance at which unions occur.

Because nearly all networks in the Kinsources dataset are disconnected, the average of a network's standard, genetic, and distance to union is typically infinite. Considering only finite distances, the following holds. In the aggregate network $A_G$, for all pairs of individuals at finite distances the average standard distance between individuals is $\langle\ell_s\rangle_A=10.97$, the average genetic distance is $\langle\ell_g\rangle_A=11.83$ and the average distance to union is $\langle\ell_u\rangle_A=7.79$. The average and standard deviation of these quantities over the individual Kinsources datasets $\{G_i\}_{i=1}^{105}$ are $\langle\ell_s\rangle_i=7.59$ and $\sigma_i(\ell_s)=3.87$, $\langle\ell_g\rangle_i=4.99$ and $\sigma_i(\ell_g)=2.50$, and $\langle\ell_u\rangle_i=5.71$ and $\sigma_i(\ell_u)=1.84$ (see Table \ref{tab:1}). 

A notable feature of these distances is that the ratios 
\[
\langle \ell_s \rangle_i/\langle \ell_s \rangle_A = 0.69, \qquad
\langle \ell_g \rangle_i/\langle \ell_g \rangle_A = 0.42, \qquad
\langle \ell_u \rangle_i/\langle \ell_u \rangle_A = 0.73,
\]
are all less than one, indicating that average distances in the individual datasets are typically smaller than those in the aggregate dataset. This discrepancy is likely influenced by the prevalence of smaller individual networks, which tend to exhibit shorter characteristic path lengths. Among these ratios, the union-distance ratio is the largest, though only slightly larger than the standard-distance ratio. This suggests that the genetic distances associated with recorded unions are comparatively stable across scales. In particular, while individual datasets are smaller and typically exhibit shorter average genetic distances overall, their mean union distances remain closer to the aggregate-network value than do average genetic distances.\\

\begin{table}
\centering
\begin{tabular}{c|c|c|c|c|c|c|c|c|c|c}
\toprule
\toprule
Quantity &$p_{non}$ & $p_{prim}$ & $p_{join}$ & $p_{dual}$ & $c_{prim}$ & $c_{join}$ & $c_{dual}$ & $\rho_{cu}$ & $\rho_{cp}$ & $\rho_{up}$\\
\midrule
$\langle\cdot\rangle_A$ &$0.52$&$ 0.05 $&$ 0.42 $&$ 0.01 $&$ 1.48 $&$ 3.09 $&$ 5.77 $&$ 0.56 $&$ -0.02 $&$ -0.14 $\\
\midrule
$\langle\cdot\rangle_i$ &$ 0.51$&$ 0.04 $&$ 0.44 $&$ 0.004 $&$ 1.84 $&$ 2.84 $&$ 5.24 $&$ 0.49 $&$ -0.17 $&$ -0.19 $\\
$\sigma_i(\cdot)$ &$ 0.1 $&$ 0.08 $&$ 0.13 $&$ 0.006 $&$ 0.89 $&$ 0.81 $&$2.20$&$ 0.22 $&$ 0.17 $&$ 0.23 $\\
\bottomrule
\bottomrule
\end{tabular}
\caption{Parent-union-child statistics and correlations of the Kinsources networks. Values are reported for the aggregate dataset $A_G = \bigsqcup_{i=1}^{105}G_i$ (first row), along with the mean $\langle\cdot\rangle$ and standard deviation $\sigma(\cdot)$ across the 105 individual Kinsources datasets $\{G_i\}_{i=1}^{105}$ (second and third rows). Columns report, from left to right, the fractions of non, primary, joint, and dual parents given by $p_{non}$, $p_{prim}$, $p_{join}$, and $p_{dual}$, respectively. Next are the corresponding average numbers of children, $c_{prim}$, $c_{join}$, and $c_{dual}$, for the primary, joint, and dual parents, respectively. Last, the correlation coefficients between an individual's number of children and unions $\rho_{cu}$, children and parents $\rho_{cp}$, and unions and parents $\rho_{up}$ are shown.}    
\label{tab:2}
\end{table}

\textbf{Parent-Union-Child Statistics and Correlations:} One purpose for investigating parent-union-child data is to characterize patterns of information loss in genealogical networks. There are two structurally identifiable situations in which information loss is guaranteed. First, an individual with fewer than two recorded parents necessarily reflects missing parental information. Second, an individual who has at least one child but no recorded partner indicates missing union information. In contrast, missing children cannot be inferred from network structure alone as the absence of descendants, for instance, may reflect either incomplete records or true childlessness.

To examine how guaranteed forms of information loss (parental and union) relate to potential information loss (children), we compute several quantities and their pairwise correlations. We classify individuals into four parent types and denote the corresponding fractions by $p_{non}$, $p_{prim}$, $p_{join}$ and $p_{dual}$: \textit{Non}-parents have no recorded children; \textit{primary} parents have at least one child for whom only one parent is recorded; \textit{joint} parents have at least one child for whom two parents are recorded; \textit{dual} parents have at least one child of each type. We let $c_{prim}$, $c_{join}$, and $c_{dual}$ denote the average number of children for primary, joint, and dual parents, respectively. 

For the aggregate dataset $A_G$, we find $p_{non}=0.52$, $p_{prim}=0.05$, $p_{join}=0.42$, and $p_{dual}=0.01$. The corresponding average numbers of children are $c_{prim}=1.48$, $c_{join}=3.09$, and $c_{dual}=5.77$ (see Table \ref{tab:2}). We observe that the average values of $p_{\mathrm{non}}, p_{\mathrm{prim}}, p_{\mathrm{join}},$ and $p_{\mathrm{dual}}$ are similar across the aggregate and individual datasets, although the variance across individual datasets is nonzero, indicating moderate cross-dataset variability. This variation becomes more apparent in the temporal analysis presented in Section~\ref{sec:hallmarks} (cf. Figure~\ref{fig:11} (left)).

Additionally, these statistics indicate a clear association between the number of recorded unions and the number of recorded children. To understand how parent, union, and child data are correlated, we define the following. Let $\rho_{cu}$ be the correlation coefficient between an individual’s number of children and number of unions, $\rho_{cp}$ be the correlation coefficient between an individual's number of children and number of recorded parents, and $\rho_{up}$ be the correlation coefficient between an individual's number of unions and number of parents. (In each case, we use the Pearson correlation coefficient.) Together, these measures quantify how different aspects of genealogical completeness co-vary across individuals.

For the aggregate dataset, we obtain $\rho_{cu} = 0.56$, $\rho_{cp} = -0.02$, and $\rho_{up} = -0.14$ (see Table~\ref{tab:2}), where summary statistics for the individual datasets $\{G_i\}_{i=1}^{105}$ are also available. As might be expected, there is a moderate positive correlation between an individual’s number of recorded children and number of recorded unions. In contrast, the correlation between the number of recorded parents and the number of children is negligible. More notably, we observe a weak negative correlation between the number of unions and the number of recorded parents, indicating that individuals with more recorded unions tend, on average, to have fewer recorded parents. We revisit this relationship later in Section \ref{sec:hallmarks}, where we use temporal analysis to further study this phenomenon. 

\section{Temporal Hallmarks of Genealogical Networks}\label{sec:data}

Genealogical data naturally admits a partial temporal ordering induced by parent–child relationships: if a directed edge $(i,j)\in E_C$, then the parent $i$ necessarily precedes the child $j$ in time ($i\prec j$).  Leveraging this partial order, we define a sequence of pseudogenerations for a genealogical network, which serves as a proxy for its underlying generational structure (see Equation \eqref{eq:pop} and Theorem \ref{feasible} in Section \ref{sec:pseudo}). These pseudogenerations form discrete layers in genealogical networks that provide an approximation of the network’s temporal evolution, even when explicit temporal data is incomplete or unavailable. We use this decomposition to analyze the time-dependence of the structural quantities introduced in Section \ref{sec:spat} and to relate their behavior to information loss, etc. in genealogical networks (see Section \ref{sec:hallmarks}).

\subsection{Temporal Methods: Pseudogenerations}\label{sec:pseudo}

The organization of graphs into temporal layers dates back to the work of \cite{Sugiyama}, which laid the foundation for hierarchical graph drawing. This framework was later extended in \cite{Gansner} through the network simplex method, forming the basis of the widely used Graphviz package. Subsequent refinements have enabled the layering of larger graphs with provably optimal edge lengths \cite{Eiglsperger}, while more general ranking methods---such as those designed for multitree datasets---have further improved efficiency and flexibility \cite{Marik17, Ruegg}.

\begin{figure}
\begin{center}
    \begin{tabular}{c}
        \begin{overpic}[scale=0.15]{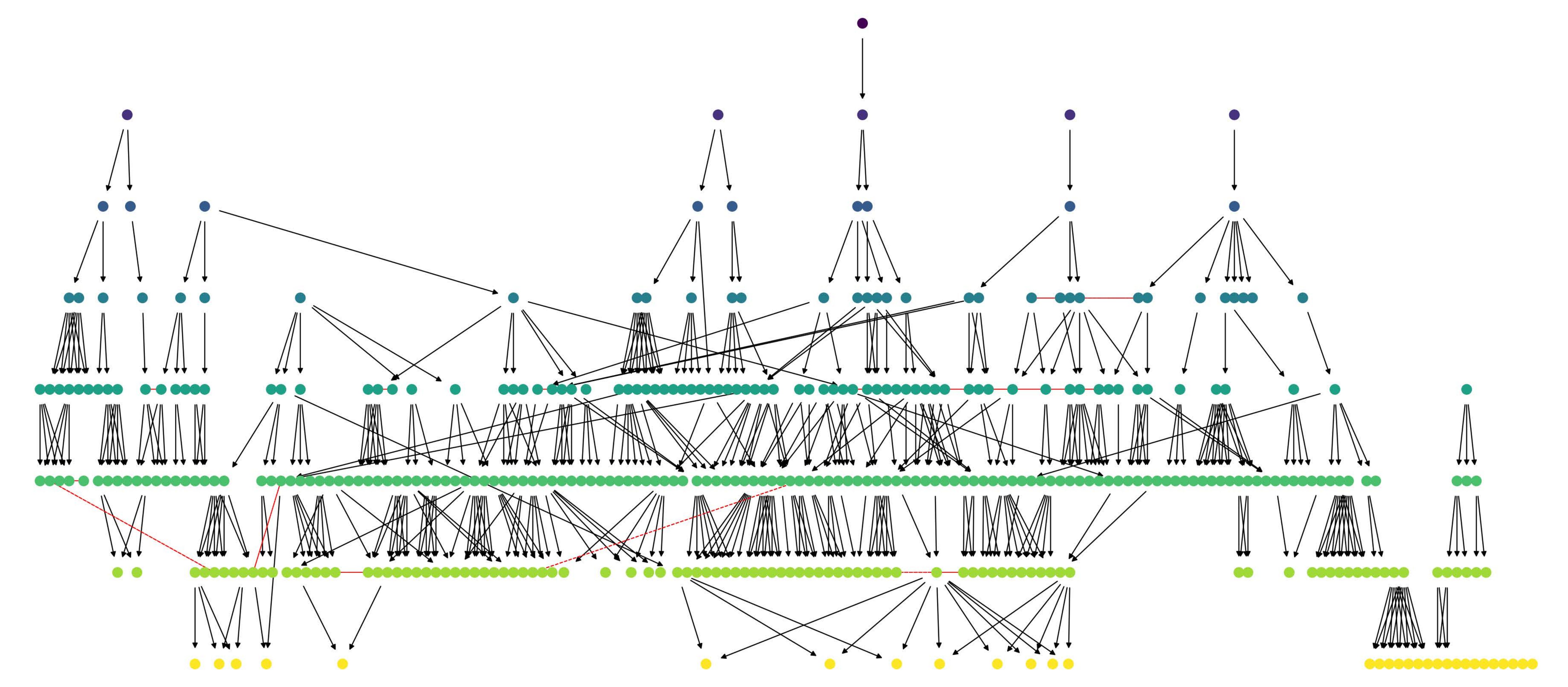}
        \put(22,-2){Layered Dogon--Konsogu--Donyu Genealogical Network}
        \end{overpic}
    \end{tabular}
\caption{The layered version of the Dogon--Konsogu--Donyu genealogical network from Figure \ref{fig:1} is shown, consisting of eight pseudogenerations generated via the pseudogeneration algorithm (see Equation \eqref{eq:pop}). As before, the edges are composed of undirected union edges, shown in red, and directed parent-child edges, shown in black. Each node's depicted color and vertical coordinate correspond to its pseudogeneration assignment.}\label{fig:7}
\end{center}
\end{figure}

Several prior approaches adapt layering techniques to genealogical networks by assigning individuals to discrete generations \cite{Bezerianos, McGuffin}. However, these methods often do not scale to large networks (e.g. those with more than $10^4$ vertices) and struggle to capture component-wide temporal patterns such as generational drift (discussed below), particularly in sparse or incomplete datasets.

To address these limitations, we introduce a new layering method formulated as an optimization problem (see Equation \eqref{eq:pop}). We prove that this problem is always feasible (Theorem~\ref{feasible}).  Although optimal assignments are not unique, the algorithm returns a deterministic pseudogeneration assignment once the node ordering and alignment-component convention are fixed (see Proposition~\ref{unique_ordering}). Moreover, as our method is posed as a constrained optimization problem, it can be readily adapted to incorporate user-specified layering preferences. 

Empirically, we characterize the runtime of the algorithm on the Kinsources datasets, which include networks larger than those considered in recent Graphviz-based approaches \cite{Marik16}. Together, these results establish pseudogenerations as a theoretically grounded framework for temporal analysis in moderate-scale genealogical networks and in larger-scale networks when approximate solutions are sufficient (cf. Figure \ref{fig:costplots}).

To describe our algorithm, let $G=(N,E_C\cup E_U)$ be a genealogical network. We denote by $C(G)=\{C_i\}_{i=1}^m$ the collection of network components (1-components) of $G$ when treated as an undirected graph. For each individual $i\in N$, let $p_i\subset N$ and $c_i \subset N$ be the set of parents and children of $i$, respectively, and let $u_i\subset E_U$ denote the set of unions incident to $i$. We define the pseudogeneration assignment as a solution to the following optimization problem.

\begin{center}
\textbf{Pseudogeneration Optimization Problem}
\end{center}


Let \(G=(N,E_C\cup E_U)\) be a genealogical network. For each \(i\in N\), let
\[
P_i=\{j\in N:(j,i)\in E_C\}, \qquad
C_i=\{j\in N:(i,j)\in E_C\}, \qquad
U_i=\{j\in N:\{i,j\}\in E_U\}.
\]
A pseudogeneration assignment is a solution of the following optimization problem:
\begin{subequations}\label{eq:pop}
\begin{alignat}{2}
\min_{g\in \mathbb{Z}_{> 0}^{N}} \quad
    & f(g)=\sum_{(i,j)\in E_C}(g_j-g_i)
    && \label{eq:pop-objective}\\
\text{subject to}\quad
    & g_j-g_i\ge 1,
    && \forall (i,j)\in E_C, \label{eq:con1}\\
    & g_k=g_\ell,
    && \forall k\in N \text{ such that } P_k=C_k=\emptyset
       \text{ and } U_k=\{\ell\}, \label{eq:con2}\\
    & g_i\in \mathbb{Z}_{> 0},
    && \forall i\in N. \label{eq:con3}
\end{alignat}
\end{subequations}
Here, $g_i$ denotes the pseudogeneration assigned to individual $i\in N$.
The objective \eqref{eq:pop-objective} minimizes total parent-child edge length.
Constraint \eqref{eq:con1} enforces temporal consistency along parent-child edges.
Constraint \eqref{eq:con2} assigns singleton union partners to the same pseudogeneration.
Constraint \eqref{eq:con3} forces pseudogenerations to be positive integers.
Note that, smaller pseudogenerations correspond to earlier generations, while larger indices correspond to more recent generations.




Below we establish several fundamental properties of the pseudogeneration optimization problem. In particular, we prove that it is always feasible under mild conditions.


\begin{theorem}\textbf{(Feasibility of the Pseudogeneration Optimization Problem)}\label{feasible}
Let $G=(N,E)$ be a genealogical network whose associated genetic network $\mathcal{G}=(N,E_C)$ is a directed acyclic graph (DAG). Then, the pseudogeneration optimization problem admits at least one feasible solution and Algorithm \ref{alg:pseudogens} constructs such a solution (see Algorithm \ref{alg:pseudogens} in the Appendix).
\end{theorem}

A proof of Theorem~\ref{feasible} is provided in the Appendix, together with pseudocode for the algorithm used to solve the pseudogeneration optimization problem in our numerical experiments. We also show that, when run to completion, the procedure yields a unique optimal solution (Proposition~\ref{unique_ordering}). This uniqueness result holds after applying the following component alignment convention.

Given an initial pseudogeneration assignment $\mathbf{g}\in\mathbb{Z}{>0}^N$, let $\omega=\max{i\in N} g_i$ denote the largest pseudogeneration in the network. Each alignment component (see Section~\ref{sec:6.2}) is then shifted by a constant amount so that its final pseudogeneration is $\omega$. This aligns the most recent known information across all network components. To remove the remaining translational ambiguity, the labeling is subsequently normalized by setting the earliest pseudogeneration equal to one; that is, $\min_{i\in N} g_i = 1$. These transformations preserve the objective value $f(\mathbf{g})$ and together constitute the \textit{component alignment convention} described in Section~\ref{sec:6.2}. An example of the resulting labeling for the Dogon--Konsogu--Donyu genealogical network is shown in Figure~\ref{fig:7}.

Algorithm~\ref{alg:pseudogens} is based on the network simplex method as implemented in Graphviz \cite{Gansner}. As with the classical simplex algorithm, its theoretical time complexity is difficult to characterize. While asymptotic bounds exist for certain primal variants \cite{Orlin, Tarjan}, these do not apply to the Graphviz implementation used here. Consequently, in the absence of known polynomial-time guarantees, we assess the algorithm’s computational complexity empirically.

Figure~\ref{fig:costplots} shows the runtime of the pseudogeneration algorithm across the 105 genealogical networks from Kinsources on a log--log scale. Runtime increases approximately monotonically with the number of individuals $n=|N|$ in a given dataset. Of these networks, 103 complete in under $T \approx 10^5$ seconds, while the two largest---the Watchi and the San Marino datasets, where $n>28\text{ }000$---exceed this threshold. For these datasets, we use approximate pseudogenerations in the subsequent temporal analysis (see Section \ref{sec:hallmarks}). We note that these approximations are empirically near-optimal: increasing the algorithm’s iterations from ten to one thousand reduces the objective by at most $0.4\%$ for these largest two networks.

\begin{figure}
\begin{center}
    \begin{tabular}{cc}
        \begin{overpic}[scale=0.3175]{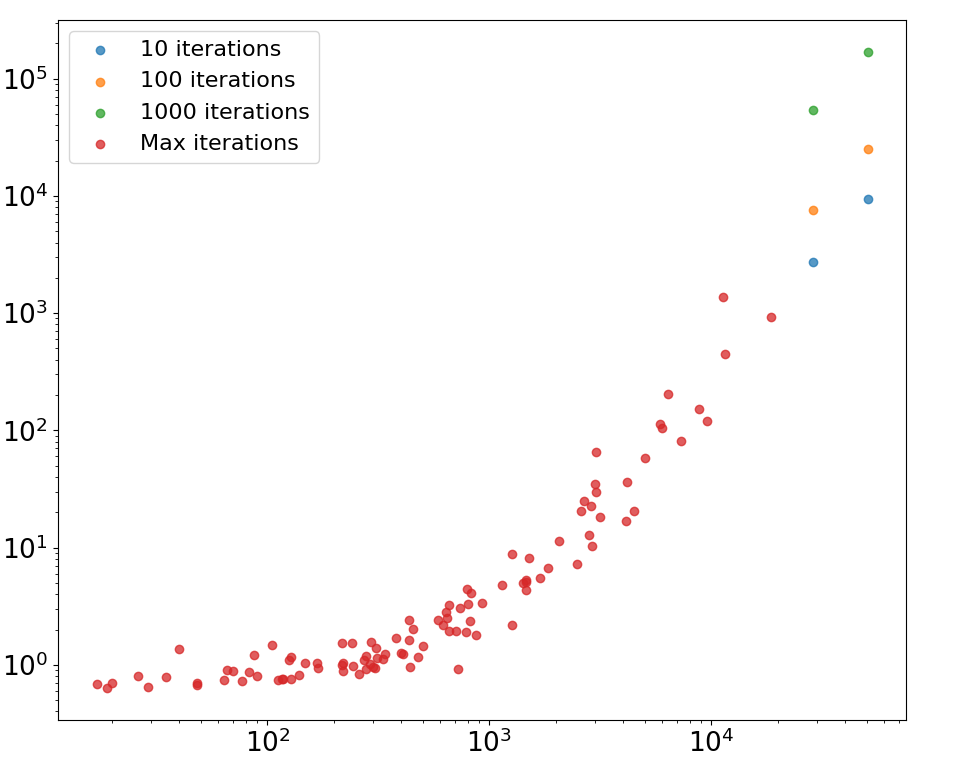}
        \put(32.75,-2.75){\scriptsize{$\mathrm{Number \ of \ Individuals}$}}
        \put(-4.5, 30){\tiny{\rotatebox{90}{$\textrm{Time (Seconds)}$}}}
        \put(17, -7.5){\scriptsize{Pseudogeneration Computational Runtime}}
        \end{overpic} 
    \end{tabular}
    \vspace{0.35cm}
\caption{Runtimes versus network size for the pseudogeneration solver using Algorithm~\ref{alg:pseudogens} (see Appendix) across all 105 Kinsources datasets $\{G_i\}_{i=1}^{105}$ are shown on a log--log scale. The algorithm converges in under $T \approx 10^5$ seconds for all but the two largest datasets (shown in red). For the two largest datasets, approximate pseudogenerations are computed using 10, 100, and 1000 internal iterations of network simplex in Algorithm~\ref{alg:pseudogens}, with corresponding runtimes indicated in blue, orange, and green, respectively.}\label{fig:costplots}
\end{center}
\end{figure}

The growth observed in Figure~\ref{fig:costplots} suggests that the empirical runtime of the pseudogeneration optimization problem increases faster than polynomially over the range of network sizes considered here. While the precise source of this computational cost remains unclear, a likely contributing factor is \textit{generational drift}---the variation in distances between lineages to a common ancestor---which complicates the determination of an optimal pseudogeneration assignment.

To highlight the flexibility of the pseudogeneration optimization problem, we refer to the procedure defined by Equation~\eqref{eq:pop} as \emph{biological sorting}, as all children of a parent are placed as close to the parents as possible. However, the optimization problem can also be readily modified to include distances between partners, or more generally to optimize a weighted combination of parent-child and union-based objectives. Such \textit{hybrid sortings} could, in principle, refine the temporal predictions presented in the remainder of the paper. Determining appropriate weights for these competing criteria, however, remains an open problem.

\subsection{Temporal Hallmarks}\label{sec:hallmarks}

Using the pseudogeneration algorithm described in the previous section, we infer an approximate temporal structure on the genealogical datasets under study. This construction enables a systematic analysis of how the structural features introduced in Section~\ref{sec:spat} vary across pseudogenerations. Our objectives in this section are twofold: (i) to characterize the temporal evolution of genealogical network structure, and (ii) to quantify how genealogical information is progressively lost when moving backward from recent to earlier pseudogenerations.

To address these objectives, we apply the pseudogeneration algorithm (see Algorithm \ref{alg:pseudogens}) to the aggregate network $A_G=(N_A,E_A)$ formed from the 105 Kinsources datasets. While alternative approaches for jointly assigning generational structure across multiple genealogical networks are possible, this algorithm provides a consistent and interpretable temporal reference frame. It allows us to track changes in network structure and genealogical completeness as a function of depth into the past, enabling a comparative analysis of structure and information loss across datasets.

For the aggregate dataset $A_G = (N_A, E_A)$, let $A_G(d) = (N_d, E_d)$ denote the restriction of $A_G$ to pseudogenerations $1$ through $d$; that is, $A_G(d)$ contains all individuals assigned to pseudogenerations at most $d$, along with all recorded edges between them. We refer to $A_G(d)$ as the depth-$d$ aggregate network, representing the first $d$ pseudogenerations of the aggregate dataset. All network statistics introduced in Section~\ref{sec:spat} can be evaluated on these depth-restricted networks as $d$ ranges from the first generation $d=1$ to the final generation (e.g. $d=40$ for the aggregate dataset).\\

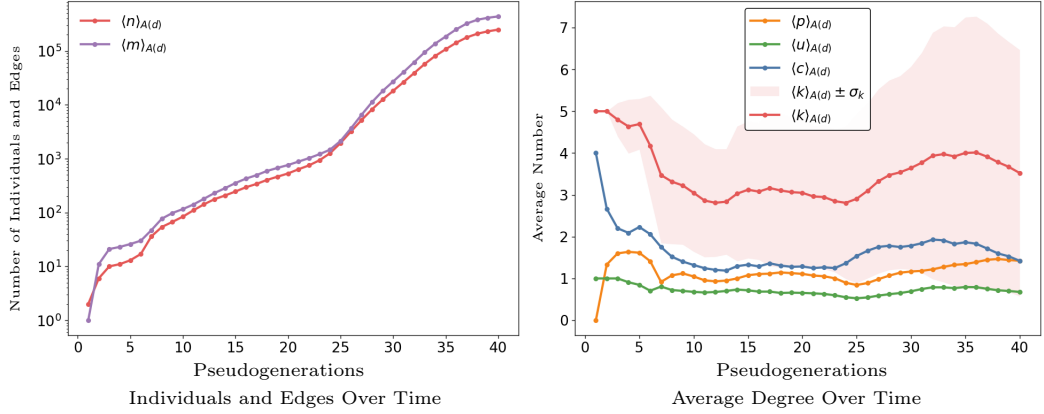
\begin{figure}
    \centering
    \begin{tabular}{cc}
        \begin{overpic}[scale=0.2000]{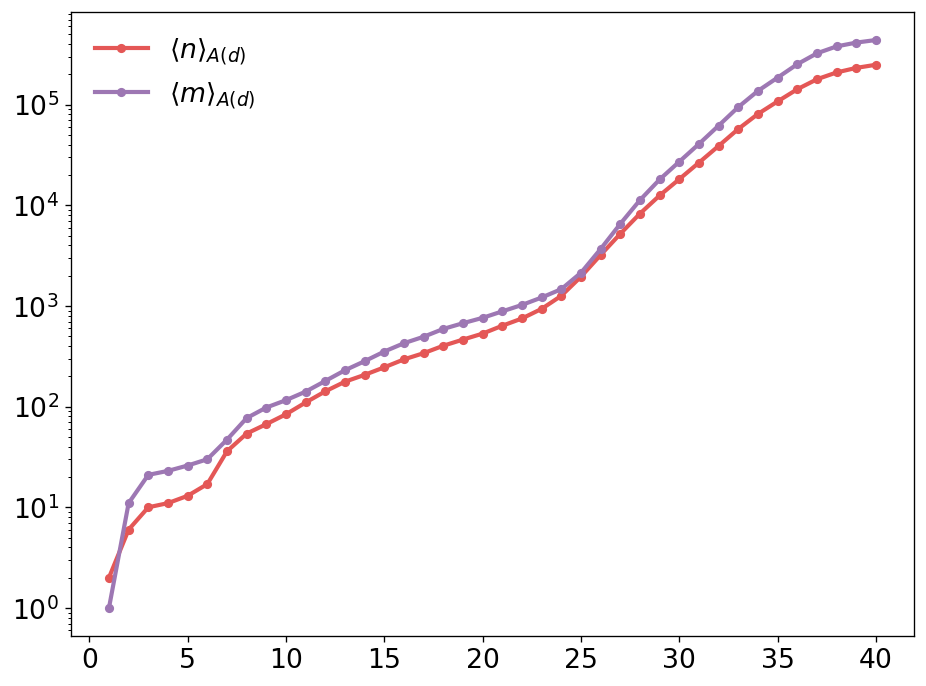}
        \put(20,-8.5){\scriptsize{Individuals and Edges Over Time}}
        \put(33,-3){\scriptsize{$\mathrm{Pseudogenerations}$}}
        \put(-4,10){\tiny{\rotatebox{90}{$\textrm{Number of Individuals and Edges}$}}}
        \end{overpic}
        \hspace{0.15cm}
        \begin{overpic}[scale=0.2000]{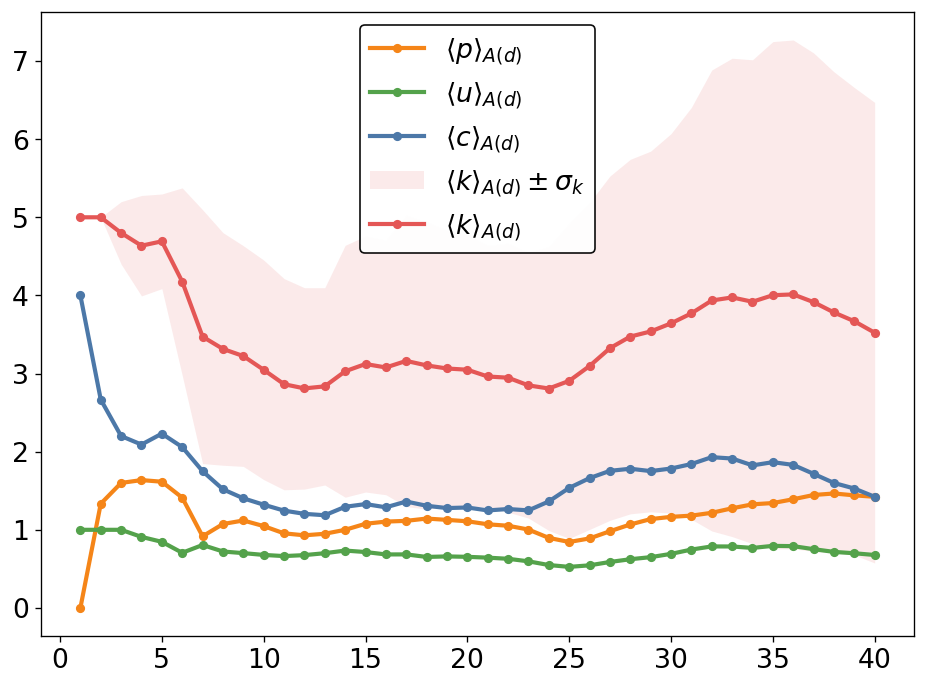}
       \put(24,-8.5){\scriptsize{Average Degree Over Time}}
        \put(33,-3){\scriptsize{$\mathrm{Pseudogenerations}$}}
        \put(-4,24){\tiny{\rotatebox{90}{$\textrm{Average Number}$}}}
        \end{overpic} 
    \end{tabular}
    \vspace{0.3cm} 
    \caption{(Left) The number of individuals $\langle n\rangle_{A(d)}$ and edges $\langle m\rangle_{A(d)}$ in the depth-$d$ aggregate network $A(d)=A_G(d)$ for $d=1,\dots,40$ are shown in red and violet, respectively. (Right) The average number of parents $\langle p\rangle_{A(d)}$, unions $\langle u\rangle_{A(d)}$, children $\langle c\rangle_{A(d)}$, and total degree $\langle k\rangle_{A(d)}$ in the depth-$d$ aggregate network $A(d)=A_G(d)$ for $d=1,\dots,40$ shown in orange, green, blue, and red, respectively. The standard deviation $\langle k\rangle_{A(d)}\pm\sigma_k$ of the average total degree $\langle k\rangle_{A(d)}$, which gives the standard deviation of the parents, unions, and children combined, is also shown in shaded red.}\label{fig:8}
\end{figure}

\textbf{Temporal Size and Degree Features:} The growth of the number of individuals and edges of the depth-$d$ aggregate datasets are shown in Figure \ref{fig:8} (left). Because Figure 8 (left) is shown on a log-linear scale, the number of individuals and edges grows approximately exponentially across pseudogenerations. The temporal evolution of the average number of parents, unions, and children for each individual in the depth-$d$ aggregate dataset is shown in Figure \ref{fig:8} (right). After an initial transient period, these quantities remain relatively stable across pseudogenerations, suggesting that the average recorded numbers of parents, unions, and children vary only modestly over most of the temporal range. Also shown are the average and standard deviation of the total degree, corresponding to the combined number of parents, unions, and children per individual. (We note that the average number of recorded parents equals the average number of recorded children in the final generation $d=40$, since both quantities are equal to the number of recorded parent-child edges divided by the total number of individuals in $A_G$.)\\  

\textbf{Temporal Assortativity and Clustering Features:} Figure \ref{fig:tempassortclust} shows the temporal evolution of the aggregate-network assortativity, genetic assortativity, and global clustering coefficient. After an initial transient period, all three quantities exhibit an overall decline: assortativity decreases from approximately $r_G \approx 0.3$ to $-0.1$, genetic assortativity from $r_{\mathcal{G}} \approx 0.3$ to $-0.16$, and the clustering coefficient from $C_G \approx 0.75$ to $0.23$. The behavior of $r_G$ and $r_{\mathcal{G}}$ suggests that family sizes are positively correlated between parents and children in earlier generations ($8<d<15$), become largely uncorrelated in intermediate generations ($15<d<25$), and eventually become negatively correlated in later generations ($25<d<40$). One possible explanation for the decline in assortativity is the increasing variance in recorded family size (total degree) over time, which is evident in Figure \ref{fig:8} (right). This suggests that increasing heterogeneity in family size may contribute to the observed reduction in assortativity. In contrast, the decrease in clustering is not simply explained by a decline in the number of unions, since Figure \ref{fig:8} (right) shows no corresponding reduction. This indicates that other mechanisms likely contribute to the observed decrease in clustering over time.\\

\begin{figure}
\centering
    \begin{tabular}{c}
        \begin{overpic}[scale=0.25]{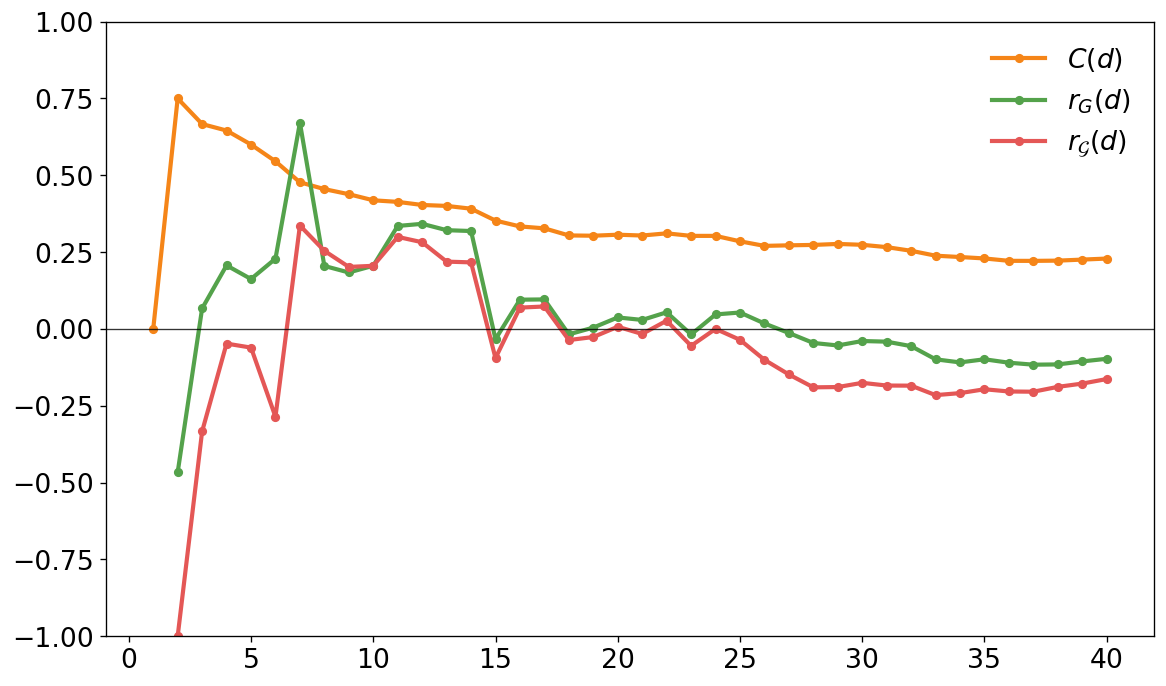}
        \put(42,-2.25){\scriptsize{{Pseudogeneration}}}
        \put(-4, 19){\scriptsize{\rotatebox{90}{$\textrm{Correlation Value}$}}}
        \put(28.5, -6.5){\scriptsize{Assortativity, and Clustering Over Time}}
        \end{overpic} 
    \end{tabular}
    \vspace{0.35cm}
\caption{Temporal evolution of the aggregate network's assortativity $r_G(d)$, genetic assortativity $r_{\mathcal{G}}(d)$, and global clustering coefficient $C_G(d)$ for $d=1,\dots,40$. After an initial transient regime, all three measures exhibit a persistent downward trend, with assortativity measures transitioning from positive to negative values in later generations (cf. Table 2).}\label{fig:tempassortclust}
\end{figure}

\textbf{Temporal Distance Features:} The average distances $\ell_s$, $\ell_g$, and $\ell_u$ for the depth-$d$ aggregate datasets are shown in Figure \ref{fig:9} (left) in blue, green, and orange, respectively. Here, pseudogenerations are indicated by dots, while the horizontal axis is given in terms of population size, allowing the average distances to be viewed as functions of the number of individuals $n$. The approximately logarithmic growth of $\ell_g$ and $\ell_u$ suggests that the aggregate network exhibits the small-world property with respect to both genetic distance and distance to union. In contrast, the standard distance $\ell_s$ does not exhibit this behavior. Instead, throughout most of the network's growth, $\ell_s$ decreases asymptotically toward $\ell_s \approx 11$. As this decrease is not observed in genetic distance, this decrease is likely driven by union edges that create shortcuts between previously distant individuals, reducing the overall average distance at least over the scales considered here (i.e. $5000<n<250\text{ }000$). Whether this decreasing trend persists indefinitely remains unknown; however, the same qualitative behavior is observed in the larger Kinsources datasets, suggesting that it is not an artifact of aggregation.

The fourth distance, $\ell_p$, shown in Figure \ref{fig:9} (left) in dashed red, is the average genetic distance from an individual to all other individuals within $\pm2$ generations. We consider this quantity because approximately $96.4\%$ of all unions in the aggregate dataset occur within this generational range. Accordingly, we refer to $\ell_p$ as the average distance to \textit{potential unions}. As shown in Figure \ref{fig:9} (left), $\ell_p$ closely tracks the overall genetic distance $\ell_g$, indicating that the average genetic distance to all individuals is well approximated by the average genetic distance restricted to individuals within $\pm2$ generations. This suggests that the dominant contribution to genetic distance arises from nearby generations, despite the much larger size of the aggregate network.

\begin{figure}
    \centering
    \begin{tabular}{cc}
        \begin{overpic}[scale=0.22]{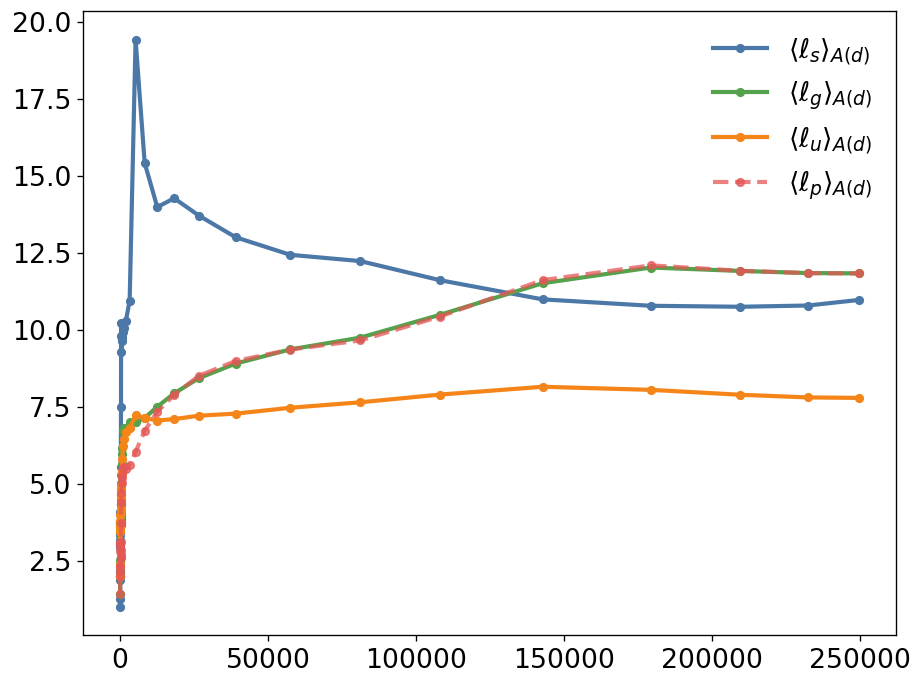}
        \put(24.5,-9){\scriptsize{Distance Over Population and Time}}
        \put(45,-3){\scriptsize{$\mathrm{Population}$}}
        \put(-3,33){\tiny{\rotatebox{90}{$\textrm{Distance}$}}}
        \end{overpic}
        \hspace{0.05cm}
        \begin{overpic}[scale=0.2235]{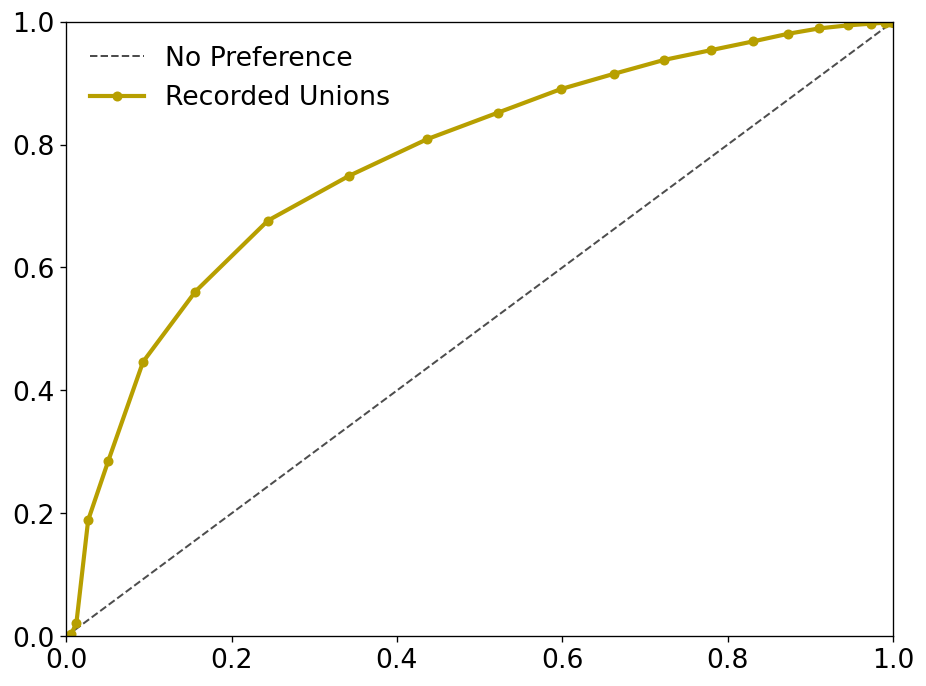}
        \put(6,-8.5){\scriptsize{Union Concentration Over Genetic Distance Percentile}}
        \put(28,-3){\scriptsize{$\mathrm{Fraction \ of \ Potential \ Unions}$}}
        \put(-3, 26){\tiny{\rotatebox{90}{$\textrm{Fraction of Unions}$}}}

                \end{overpic} 
        \end{tabular}
    \vspace{0.45cm} 
    \caption{(Left) The average of the standard $(\ell_s)$, genetic $(\ell_g)$, union $(\ell_u)$, and potential union $(\ell_p)$ distance in the depth-$d$ network $A(d)$ are plotted in blue, green, orange, and dashed red, respectively. The potential union distance is the average genetic distance from an individual to all other individuals within $\pm2$ generations. Here, the pseudogenerations $d=1,\dots,40$ are represented as dots. The genetic distance $\ell_g$ and union distance $\ell_u$ exhibit approximately logarithmic growth, indicating a small-world phenomenon, while the standard distance $\ell_s$ decreases toward a limiting value over the observed range. (Right) The fraction of unions (vertical axis) that form within a given fraction of the closest potential unions (horizontal axis) is shown in yellow. This curve of recorded unions lies well above the dashed baseline ($y=x$), indicating a strong preference for short genetic distances in union formation.}
    \label{fig:9}
\end{figure}
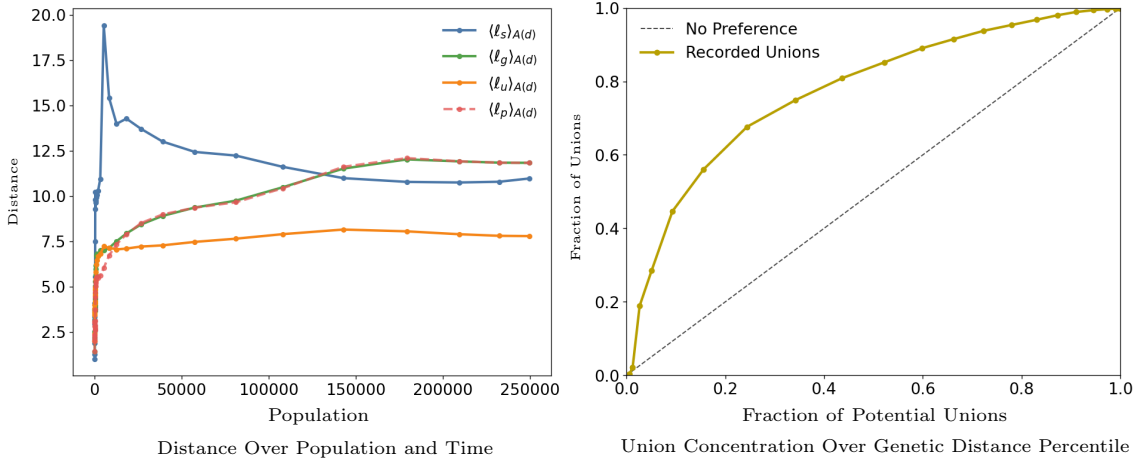         

As the pseudogeneration increases, the average distance to union and the average distance to potential unions approach $\ell_u \approx 8$ and $\ell_p \approx 12$, respectively. This indicates that unions tend to occur at substantially shorter genetic distances than the majority of potential unions. This preferential formation of unions at short genetic distances is illustrated in Figure \ref{fig:9} (right), where the cumulative fraction of recorded unions (vertical axis) is plotted against the percentile of the genetically closest potential unions (horizontal axis). The dashed diagonal line ($y=x$) represents the null model of no preference, corresponding to unions forming uniformly among potential pairs. Because the empirical curve lies substantially above this baseline, unions are strongly concentrated among genetically close individuals. For example, the point $(0.1563,0.5)$ indicates that $50\%$ of all recorded unions occur within the genetically closest $15.63\%$ of potential pairings in the Kinsources dataset.\\

\textbf{Temporal Parental Features:} The fractions of non-parents, primary parents, joint parents, and dual parents in the depth-$d$ aggregate network are shown in Figure~\ref{fig:11} (left) as functions of pseudogeneration depth $d$. The proportion of non-parents exhibits a generally increasing trend as later pseudogenerations are included, rising from zero percent in the earliest depth-restricted network to 52\% in the full aggregate network. In contrast, the fraction of joint parents shows the opposite behavior, decreasing from comprising all individuals in the earliest depth-restricted network to 42\% in the full aggregate network. The fraction of primary parents has a somewhat unimodal shape, beginning and ending near zero and peaking near 40\%. The fraction of dual parents remains nearly constant at roughly 0.01. 

While the reasons for these trends are currently unknown, one possible contributing factor is the alignment-component convention in the pseudogeneration algorithm, which shifts alignment components toward later pseudogenerations. These components could, in principle, contain relatively fewer joint parents and more non-parents which, if true, would strengthen the idea that component size is correlated with information loss in genealogical networks. Similarly, the unimodal distribution of primary parents may reflect a transition in data completeness over time: in earlier generations, partially observed parent information (i.e. primary parents) may degrade into missing information where neither parent is known, while in later generations this partial information may instead resolve into more complete records in which both parents are identified (i.e. joint parents).\\ 

\begin{figure}
\begin{center}
    \begin{tabular}{ccc}
        \begin{overpic}[scale=0.2200]{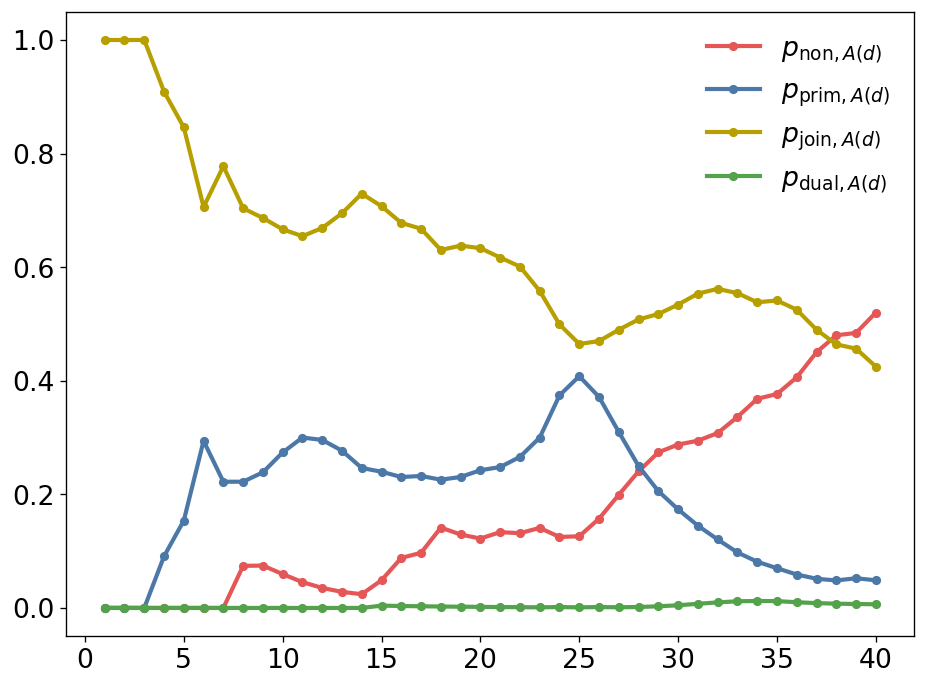}
        \put(32.5,-8.5){\scriptsize{Parent Type Over Time}}
        \put(38,-3){\scriptsize{$\mathrm{Pseudogenerations}$}}
        \put(-4,21){\tiny{\rotatebox{90}{$\textrm{Fraction of Parent Type}$}}}
        \end{overpic} \hspace{0.05cm}
        \begin{overpic}[scale=0.2200]{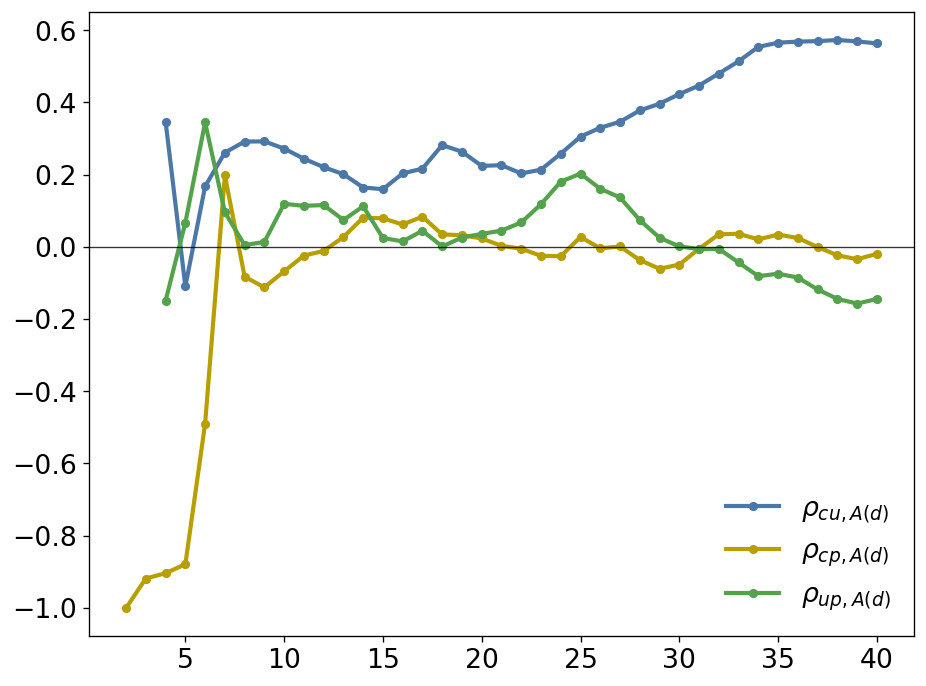}
        \put(24,-8.5){\scriptsize{Correlation Coefficients Over Time}}
        \put(38,-3){\scriptsize{$\mathrm{Pseudogenerations}$}}
        \put(-2,31){\tiny{\rotatebox{90}{$\textrm{Correlation}$}}}
        \end{overpic}
    \end{tabular}
    \vspace{0.35cm}
\caption{(Left) The fraction of non-parents ($p_{non}$), primary parents ($p_{prim}$), joint parents ($p_{join}$), and dual parents ($p_{dual}$) in the depth-$d$ network $A(d)$ are shown as functions of pseudogeneration depth $d$ in red, blue, yellow, and green, respectively. (Right) The correlations between number of recorded unions and children ($\rho_{cu}$), children and parents ($\rho_{cp}$), and unions and parents ($\rho_{up}$) in the depth-$d$ network $A(d)$ are shown as functions of pseudogeneration depth $d$ in blue, yellow, and green, respectively.}\label{fig:11}
\end{center}
\end{figure}
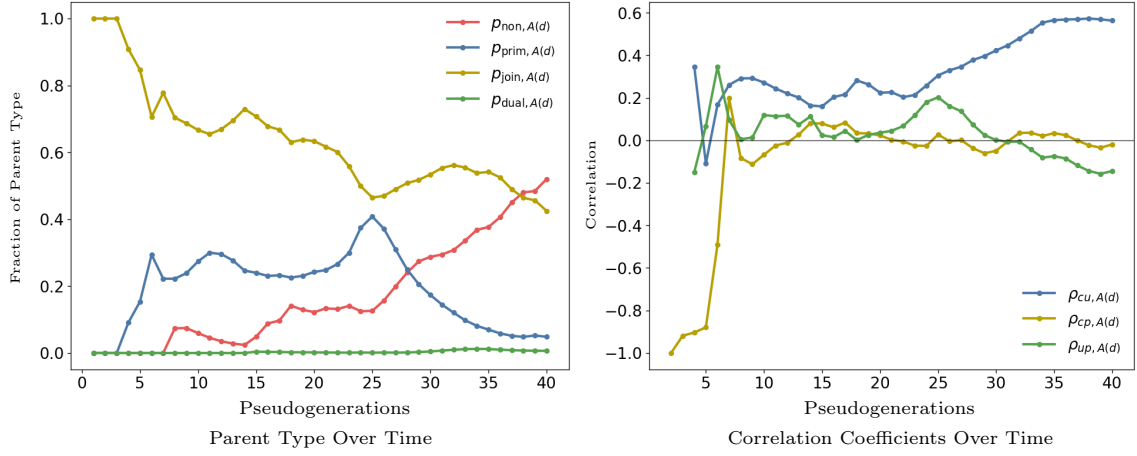

\textbf{Temporal Parent-Union-Child Correlations:} For the aggregate dataset, the correlations $\rho_{cu}$, $\rho_{cp}$, and $\rho_{up}$ are shown per pseudogeneration in Figure~\ref{fig:11} (right), wherein each exhibits a distinct trend. The correlation $\rho_{cu}$ between the number of children and the number of unions is generally positive and increases after generation 20, plateauing at approximately $\rho_{cu} \approx 0.56$ in the final generations. In contrast, the correlation $\rho_{cp}$ between the number of children and the number of parents fluctuates near zero beyond generation 10. Meanwhile, the correlation $\rho_{up}$ between the number of unions and the number of parents declines from roughly $\rho_{up} \approx 0.2$ around generation 25 to nearly $\rho_{up}\approx-0.2$ in the final generations.

The trend in $\rho_{cu}$ suggests that, over time, a greater number of recorded unions is associated with a greater number of recorded children, and vice versa. In contrast, the near-zero values of $\rho_{cp}$ across most pseudogenerations indicate that information about an individual’s parents provides little predictive value for the number of their children. One possible explanation for this weak relationship is that $\rho_{cp}$ links information separated by two generations---an individual's recorded parents and recorded children---whereas $\rho_{cu}$ and $\rho_{up}$ involve quantities typically separated by a single generation. The behavior of $\rho_{up}$ is more surprising, exhibiting an approximately monotonic decline from moderately positive to moderately negative values over the latter half of the pseudogenerations. The cause of this trend remains unclear; in particular, it is not yet understood why increased information regarding parents is associated with fewer recorded unions in later pseudogenerations.   

\section{Conclusion}\label{sec:conc}

In this paper, we develop a unified framework for the comparative analysis of genealogical networks that integrates structural and temporal perspectives through the notion of pseudogenerations. By analyzing more than one hundred genealogical datasets from the Kinsources repository, which span various cultures, historical periods, and recording practices, our intent is to move beyond individual case studies to investigate features that recur across human kinship networks.

Using this framework, we identify a collection of common structural and temporal hallmarks. These include scale-free--like degree and component-size distributions, multiscale family organization, small-world behavior with respect to genetic and union-based distances, consistent degree disassortativity, preferential union formation at short genetic distances, and recurring temporal and demographic patterns in genealogical records. We further show that 2-components provide a natural unit of genealogical structure, balancing connectivity and interpretability while preserving meaningful family relationships. Together, these findings provide evidence that diverse genealogical datasets share common organizational features despite substantial variation in cultural and historical context.

Beyond these empirical findings, the introduction of pseudogenerations enables the temporal analysis of otherwise static genealogical data. By extracting an inferred generational structure directly from network topology, this approach makes it possible to study how network properties evolve across generations and provides a unified method for integrating structural and temporal perspectives on kinship networks.

These results provide a foundation for the development of general models of genealogical networks. The recurrence of structural and temporal features across diverse datasets suggests there are organizing principles that persist across genealogical networks. More broadly, the framework introduced here creates opportunities for future work in comparative genealogy, predictive modeling, population dynamics, and the study of social organization through kinship structure. 

Building on this foundation, a central objective is to develop a quantitative theory of information loss in genealogical networks. Such a theory would make it possible to infer where genealogical information is missing, estimate how much information has been lost, and characterize the processes responsible for the loss. By linking observed network structure to patterns of missing information, this approach has the potential to improve the interpretation and reconstruction of incomplete genealogical records. Progress in this direction will be reported in forthcoming work.

\section{Acknowledgements} The authors used generative artificial intelligence tools to assist with language editing, phrasing, and improvement of manuscript readability, as well as to help review and debug computer code used in this research.

\section{Appendix}\label{sec:app}

Here, we give the rigorous statements and proofs of the results referenced in this paper. The first set of results are related to a genealogical network’s $k$-components (see Section \ref{sec:spat}). The second set of results refer to the pseudogeneration optimization problem (see Section \ref{sec:data}). 

\subsection{Structure of $k$-Components}
To characterize when $k$-components group or separate parents and children, we restrict our analysis to a class of genealogical networks with well-defined union structure. Specifically, a genealogical network $G=(N,E_U\cup E_C)$ is said to be \emph{union-complete} if:

(i) every union edge $\{i,j\} \in E_U$ has at least one recorded child $c\in N$, so that $(i, c), (j,c) \in E_C$;

(ii) every individual has either zero or two recorded parents;

(iii) whenever $i$ and $j$ are the two recorded parents of a child $c$, the union edge \(\{i, j\} \in E_U\) is present.

We note that although the aggregate dataset $A_G$ is not strictly union-complete, it is close to satisfying this property: approximately 72\% of partner pairs have a recorded child, and only 8.2\% of individuals have a single recorded parent.

In the union-complete setting, the following hold.

\begin{lemma}
\label{lemma:pc_edges}
Let $G=(N,E_U\cup E_C)$ be a union-complete genealogical network.  Then, every edge of $G$ belongs to a triangle in the underlying undirected graph consisting of one union edge and two parent-child edges.
\end{lemma}
\begin{proof}

Let \(e\in E_U\cup E_C\). 

If \(e=\{i,j\}\in E_U\), then by union-completeness there exists a recorded child \(c\) of \(i\) and \(j\). Hence \(\{i,j\},(i,c),(j,c)\) form a triangle.

If \(e=(i,c)\in E_C\), then \(c\) has exactly two recorded parents, say \(i\) and \(j\), by union-completeness. By the convention that recorded co-parents are joined by a union edge, \(\{i,j\}\in E_U\). Hence again, \(\{i,j\},(i,c),(j,c)\) form a triangle.
\end{proof}


\begin{proposition}\textbf{(Parent-Child Separation)}
\label{thm:3-connected}
Let $G=(N,E_U\cup E_C)$ be a union-complete genealogical network whose associated genetic network $\mathcal{G}$ is a DAG.\\
(i) If $k=1$ or $k=2$, every edge of $G$ belongs to exactly one $k$-component of $G$, and every individual with two recorded parents shares a $k$-component with both recorded parents.\\
(ii) If $k\geq 3$, any $k$-component of $G$ contains an individual whose child lies outside the $k$-component. 
\end{proposition}
\begin{proof}
    Part (i) is trivial for $k=1$. For $k=2$,
    note by the block decomposition of $G$ (e.g. see \cite{diestel2025graph}), every edge of $G$ is a bridge edge or belongs to a unique 2-component of $G$. Since every edge of $G$ belongs to a cycle (a triangle, by Lemma \ref{lemma:pc_edges}), no bridge edges exist, and thus each edge belongs to a unique 2-component. Furthermore, every individual with two recorded parents forms a triangle with those parents, and hence shares a 2-connected subgraph with both recorded parents.  This subgraph is contained in a unique 2-component.
    
    
We now prove part (ii). Let $S$ be a $k$-component of $G$ with $k \geq 3$. Then, $S$ is 3-connected.  If $S$ contains no parent-child edge, then any union edge in $S$ has, by union-completeness, a child outside $S$, so the claim follows.  Thus, we assume $S$ contains a parent-child edge.  Therefore, the set of child nodes with a parent in $S$,
\[C = \{c \in N : \exists p \in S \text{ such that } (p, c) \in E_C\},\] is nonempty.

Now suppose, by way of contradiction, that every individual in $C$ is contained in $S$.  Choose $n_0 \in C$.  Since $n_0$ has a recorded parent and $G$ is union-complete, $n_0$ has exactly two recorded parents and, hence, two incoming parent-child edges.   Since $S$ is 3-connected, $n_0$ has degree at least three in $S$.  Hence, there must be an additional edge in $S$ incident to $n_0$ that is not one of these incoming parent-child edges. 

We then have two cases:
    
(Case 1) Assume $\{n_0,i_1\} \in E_U$ is a union edge. Then the pair $n_0$ and $i_1$ has a child $i_2 \in N$ by union-completeness.  Since $n_0 \in S$, we have $i_2 \in C$, and by the contradiction hypothesis, $i_2 \in S$.  Therefore, $i_2\in S$ is a child of $n_0$. In this case set $n_1=i_2$.

(Case 2) Assume the additional edge is a parent-child edge. Since the additional edge is not an incoming parent edge, it must be outgoing.  Thus, we assume $(n_0,i_1)\in E_C$ is an outgoing parent-child edge. Then $i_1\in S$ is a child node of $n_0$. In this case set $n_1=i_1$.
    
In either case, we obtain a parent-child edge $(n_0,n_1)\in E_C$ contained in $S$ with $n_1 \in C \cap S$. Since $n_1 \in C \cap S$, the same argument applies to $n_1$.  Repeating this process produces an additional edge $(n_1,n_2)\in E_C$.  After $|S|$ iterations, we obtain a sequence of edges
    \[(n_0,n_1),(n_1,n_2),\dots,(n_{|S|-1},n_{|S|})\in E_C\]
which form a directed walk in the associated genetic network $(N, E_C)$. Since $\mathcal{G}=(N, E_C)$ is a DAG, the vertices $(n_0, \dots, n_{|S|})$ are pairwise distinct.  After $|S|$ iterations, we obtain $|S| + 1$ pairwise distinct vertices in $S$, contradicting the fact that $S$ has exactly $|S|$ vertices.  This contradiction proves part (ii).
\end{proof}

\subsection{The Pseudogeneration Optimization Problem}\label{sec:6.2}
We now consider the pseudogeneration optimization problem (POP) introduced in Section~\ref{sec:data} (cf. Equation~\eqref{eq:pop}). Algorithm~\ref{alg:pseudogens} summarizes the procedure, which constructs a constrained hierarchical layout of the associated genetic network $\mathcal{G}=(N,E_C)$.  The algorithm first solves the rank-assignment problem with singleton-union rank constraints, then computes alignment components, selects a terminal vertex in each alignment component, and performs a second network simplex pass in which these terminal vertices are constrained to share a common rank.

For the datasets we consider, $\mathcal{G}$ is a directed acyclic graph, allowing pseudogenerations to be computed via the network simplex method. The algorithm incorporates constraints arising from \emph{singleton unions}, which are unions involving individuals with no parents or children, by assigning such partners to share the same pseudogeneration. 

Since union edges do not encode parent-child temporal precedence, this convention is applied not to the connected components of the full genealogical network, but to the connected components of an auxiliary alignment graph determined by parent-child structure together with the singleton-union constraints.  Specifically, let $E_U^{\text{sing}} \subseteq E_U$ denote the set of union edges $\{k, \ell\}$ appearing in constraint (\ref{eq:con2}), and define the auxiliary alignment graph
\[H = (N, E_C^{\text{und}}\cup E_U^{\text{sing}}),\]
where $E_C^{\text{und}}$ denotes the undirected version of the parent-child edge set.  We refer to the connected components of $H$ as \textit{alignment components}. In addition, each alignment component is \textit{anchored} by fixing at least one of its individuals in the latest (i.e.  largest) pseudogeneration.

\begin{algorithm}
\caption{Pseudogeneration Assignment via Graphviz}
\label{alg:pseudogens}
\begin{algorithmic}
\Require Genealogical network $G=(N,E_C\cup E_U)$ whose genetic network $(N,E_C)$ is a DAG.
\Ensure Pseudogeneration assignment $g:N\to\mathbb Z_{>0}$.

\Procedure{Pseudogenerations}{G}
\State $A\gets \texttt{to\_agraph}(G)$. \Comment{Create Graphviz object.}

\ForAll{\(\{u,v\}\in E_U\)}
    \State Set the Graphviz edge attribute \texttt{constraint=false}.
\EndFor

\State \(E_U^{\mathrm{sing}}\gets \emptyset\).

\ForAll{\(k\in N\) such that \(P_k=C_k=\emptyset\) and \(U_k=\{\ell\}\)}
    \State \(E_U^{\mathrm{sing}}\gets E_U^{\mathrm{sing}}\cup\{\{k,\ell\}\}\).
    \State Add a Graphviz subgraph containing \(k\) and \(\ell\) with attribute \texttt{rank=same}. \Comment{Enforces \((3c)\).}
\EndFor

\State Set Graphviz attributes \texttt{TBbalance}, \texttt{newrank}, and iteration limits as specified.

\State \(A^{(1)}\gets \texttt{NetworkSimplex}(A)\). \Comment{First network-simplex pass.}
\State Let \(g^{(1)}\) be the rank assignment returned by \(A^{(1)}\).

\State Define the auxiliary alignment graph \(H=(N,E_C^{\mathrm{und}}\cup E_U^{\mathrm{sing}})\).
\State Compute the connected components \(C_1,\ldots,C_q\) of \(H\).

\ForAll{\(C_w\), \(w=1,\ldots,q\)}
    \State \(v_w^*\gets \arg\max_{v\in C_w} g_v^{(1)}\). \Comment{Select a terminal vertex in \(C_w\).}
\EndFor

\State Add a Graphviz subgraph containing \(v_1^*,\ldots,v_q^*\) with attribute \texttt{rank=same}.

\State \(A^{(2)}\gets \texttt{NetworkSimplex}(A)\). \Comment{Second network-simplex pass.}
\State Let \(g\) be the rank assignment returned by \(A^{(2)}\).

\State Normalize \(g\) by setting \(g_i\gets g_i-\min_{j\in N}g_j+1\) for all \(i\in N\).

\State \Return \(g\).

\EndProcedure
\end{algorithmic}
\end{algorithm}

Since Graphviz does not natively support deepest-generation anchoring, we implement the network simplex method in two stages: first to compute a provisional layering, and then again after fixing a deepest individual in each alignment component to enforce the anchoring constraint.  The alignment step in Algorithm \ref{alg:pseudogens} implements the alignment-component convention by selecting a terminal vertex in each connected component of the auxiliary alignment graph $H=(N, E_C^{\text{und}} \cup E_U^{\text{sing}})$ and constraining these selected vertices to share a common rank before the second network simplex pass.  This enforces the same alignment as uniformly shifting each alignment component so that its terminal pseudogeneration agrees with the global terminal pseudogeneration.

We now present a proof of Theorem~\ref{feasible}, which guarantees the feasibility of the POP.

\begin{proof}
Because the associated genetic network $\mathcal{G} = (N, E_C)$ is a finite DAG, a topological ordering exists. We therefore choose a topological sorting and define $h_j$ to be the position of $j$ in this ordering, beginning at 1.  Then, for each parent-child edge $(i,j) \in E_C$, $i$ precedes $j$, so $h_i < h_j$.  Hence, $h_j - h_i \geq 1$, so the preliminary assignment $h$ satisfies constraint \ref{eq:con1}.  Further, by construction, $h_j \in \mathbb{Z}_{> 0}$, so $h$ satisfies constraint \ref{eq:con3} as well.

It remains to satisfy constraint \ref{eq:con2}.  This constraint applies only to individuals $k$ with no recorded parents, no recorded children, and exactly one recorded union partner $\ell$.  Initialize $g_i = h_i$ for all $i \in N$.  Then, for each individual $k$ satisfying constraint \ref{eq:con2}, set $g_k = g_\ell$, where $U_k = \{\ell\}$, unless both endpoints satisfy this constraint, in which case assign both endpoints the same positive integer value.  Because $k$ has no incident parent-child edges, changing its pseudogeneration does not alter any parent-child difference $g_j - g_i$.  Thus, the temporal constraints remain satisfied, the singleton-union constraint is satisfied by construction, and all pseudogeneration values remain positive integers.

This leaves a nonempty integer feasible set.  Moreover, the objective is bounded below by $|E_C|$, since $g_j - g_i \geq 1$ for every $(i, j) \in E_C$.  Because the objective takes integer values, an optimal feasible assignment exists.

Algorithm \ref{alg:pseudogens} applies network simplex to the rank-assignment linear program corresponding to Equation (\ref{eq:pop}).  Since the feasible region is nonempty and the objective is bounded below, network simplex is guaranteed to reach an optimal solution when run to convergence (\cite[Ch.~11, \S11.5]{Ahuja}).  The constraint matrix has at most two nonzero entries in each difference-constraint row, one $+1$ and one $-1$, and is the node-arc incidence matrix of an auxiliary directed graph, possibly augmented by identity rows for lower bounds.  Hence, the matrix is totally unimodular.  Because the right-hand side is integral, the linear-programming relaxation has integral optimal vertices, so an optimal integral solution exists and can be returned by network simplex.  Since this is the same rank-assignment formulation used by Graphviz's version of network simplex (\cite{Gansner}), Algorithm \ref{alg:pseudogens} returns an optimal solution to the pseudogeneration optimization problem when the network simplex routine is run to convergence.

Finally, the post-processing alignment convention preserves feasibility and optimality.  Indeed, shifting every pseudogeneration value in an alignment component by the same constant leaves each parent-child difference $g_j - g_i$ unchanged, since both endpoints of any parent-child edge lie in the same alignment component.  Therefore, constraint \ref{eq:con1} and the objective function are preserved.  Constraint \ref{eq:con2} is also preserved because the endpoints of each singleton-union edge in $E_U^{\text{sing}}$ lie in the same alignment component and are shifted by the same constant.  Lastly, the alignment shifts each alignment component forward so that its maximal pseudogeneration agrees with the global maximal pseudogeneration; hence, all shifts are nonnegative and all pseudogeneration values remain positive, preserving constraint \ref{eq:con3}.  Thus, the aligned assignment remains feasible and optimal.

The uniformly shifted assignment therefore provides a feasible solution to the second-pass constrained problem with the same objective value as the original optimum.  Hence, the second network simplex pass also attains the same optimal objective value when run to convergence.


\end{proof}

Although Algorithm~\ref{alg:pseudogens} always attains the minimum of the objective function $f(\mathbf{g})$ when run to convergence, the corresponding pseudogeneration labeling $\mathbf{g}$ need not be unique: distinct assignments may achieve the same optimal value. However, for a fixed node ordering, the network simplex method follows a deterministic sequence of operations and converges to a specific optimal solution. Consequently, the resulting pseudogeneration assignment is reproducible across runs. This yields the following result.

\begin{proposition}\textbf{(Order-Deterministic Optimality)}\label{unique_ordering}
Given a fixed node ordering and deterministic tie-breaking in the network simplex routine, the pseudogeneration algorithm produces a deterministic optimal solution to the pseudogeneration optimization problem.
\end{proposition}
\begin{proof}

Given a fixed node ordering and deterministic tie-breaking in the network simplex routine, each step of Algorithm \ref{alg:pseudogens} is deterministic:  the initial rank assignment, the singleton-union constraints, the alignment-component anchoring step, and the final network simplex pass.  Hence, the algorithm returns the same optimal pseudogeneration assignment on repeated runs.

\end{proof}

{\footnotesize

\begin{longtable}{|>{\raggedright\arraybackslash}p{0.28\textwidth}|r|r|>{\raggedright\arraybackslash}p{0.47\textwidth}|}\caption{Summary of the individual Kinsources datasets $\{G_i\}_{i=1}^{105}$ including number of nodes, edges, and Kinsources URLs.} \\
\hline
\hline
\textbf{Network Name} & \textbf{Nodes} & \textbf{Edges} & \textbf{Kinsources URL} \\
\hline
\endfirsthead

\hline
\textbf{Network Name} & \textbf{Nodes} & \textbf{Edges} & \textbf{Kinsources URL} \\
\hline
\endhead

\hline
\endfoot

family & 17 & 24 & https://www.kinsources.net/kidarep/dataset-31-family.xhtml\\ \hline
ngatatjara\_1966\_au04 & 19 & 30 & https://www.kinsources.net/kidarep/dataset-21-ngatatjara-1966-au04.xhtml\\ \hline
wanindiljaugwa\_1948\_au06 & 20 & 29 & https://www.kinsources.net/kidarep/dataset-38-wanindiljaugwa-1948-au06.xhtml\\ \hline
natchez\_1700 & 26 & 17 & https://www.kinsources.net/kidarep/dataset-40-natchez-1700.xhtml\\ \hline
hatfields\_and\_mccoys & 29 & 48 & https://www.kinsources.net/kidarep/dataset-46-hatfields-and-mccoys.xhtml\\ \hline
gundangborn\_1948\_au02 & 35 & 53 & https://www.kinsources.net/kidarep/dataset-81-gundangborn-1948-au02.xhtml\\ \hline
angmagsalik\_1884\_nu01 & 40 & 59 & https://www.kinsources.net/kidarep/dataset-33-angmagsalik-1884-nu01.xhtml\\ \hline
hare\_1956\_nd05 & 48 & 76 & https://www.kinsources.net/kidarep/dataset-84-hare-1956-nd05.xhtml\\ \hline
vedda\_1905\_as04 & 48 & 86 & https://www.kinsources.net/kidarep/dataset-41-vedda-1905-as04.xhtml\\ \hline
takamiut\_1927\_64\_nu03 & 64 & 109 & https://www.kinsources.net/kidarep/dataset-91-takamiut-1927-64-nu03.xhtml\\ \hline
kutchin\_1947\_nd06 & 66 & 109 & https://www.kinsources.net/kidarep/dataset-82-kutchin-1947-nd06.xhtml\\ \hline
belcher\_island\_1958\_nu04 & 70 & 102 & https://www.kinsources.net/kidarep/dataset-50-belcher-island-1958-nu04.xhtml\\ \hline
slavey\_1911\_nd12 & 77 & 134 & https://www.kinsources.net/kidarep/dataset-69-slavey-1911-nd12.xhtml\\ \hline
semang\_1924\_50\_as03 & 83 & 127 & https://www.kinsources.net/kidarep/dataset-8-semang-1924-50-as03.xhtml\\ \hline
tiwi & 87 & 111 & https://www.kinsources.net/kidarep/dataset-216-tiwi.xhtml\\ \hline
miwuyt\_1967\_au03 & 90 & 119 & https://www.kinsources.net/kidarep/dataset-12-miwuyt-1967-au03.xhtml\\ \hline
arara & 105 & 246 & https://www.kinsources.net/kidarep/dataset-87-arara.xhtml\\ \hline
oodnadatta & 112 & 182 & https://www.kinsources.net/kidarep/dataset-15-oodnadatta.xhtml\\ \hline
jie & 116 & 177 & https://www.kinsources.net/kidarep/dataset-226-jie.xhtml\\ \hline
eyak\_1890 & 118 & 192 & https://www.kinsources.net/kidarep/dataset-39-eyak-1890.xhtml\\ \hline
tlingit & 125 & 202 & https://www.kinsources.net/kidarep/dataset-242-tlingit.xhtml\\ \hline
kaingang & 128 & 123 & https://www.kinsources.net/kidarep/dataset-164-kaingang.xhtml\\ \hline
shoshone\_1880\_nd11 & 128 & 202 & https://www.kinsources.net/kidarep/dataset-23-shoshone-1880-nd11.xhtml\\ \hline
paiute\_1880\_nd09 & 139 & 201 & https://www.kinsources.net/kidarep/dataset-79-paiute-1880-nd09.xhtml\\ \hline
pul\_eliya\_1954\_simpler\_version & 147 & 243 & https://www.kinsources.net/kidarep/dataset-78-pul-eliya-1954-simpler-version.xhtml\\ \hline
labrador\_inuit\_1776\_nu02 & 168 & 221 & https://www.kinsources.net/kidarep/dataset-14-labrador-inuit-1776-nu02.xhtml\\ \hline
top\_of\_the\_mountain & 169 & 275 & https://www.kinsources.net/kidarep/dataset-11-top-of-the-mountain.xhtml\\ \hline
port\_keats & 216 & 286 & https://www.kinsources.net/kidarep/dataset-254-port-keats.xhtml\\ \hline
ainu\_1880\_as01 & 216 & 378 & https://www.kinsources.net/kidarep/dataset-22-ainu-1880-as01.xhtml\\ \hline
lainiovouma\_1952\_eu03 & 218 & 353 & https://www.kinsources.net/kidarep/dataset-17-lainiovouma-1952-eu03.xhtml\\ \hline
suya & 219 & 371 & https://www.kinsources.net/kidarep/dataset-171-suya.xhtml\\ \hline
tikar & 240 & 395 & https://www.kinsources.net/kidarep/dataset-158-tikar.xhtml\\ \hline
waimiri-atroari & 244 & 482 & https://www.kinsources.net/kidarep/dataset-66-waimiri-atroari.xhtml\\ \hline
konkomba & 258 & 253 & https://www.kinsources.net/kidarep/dataset-206-konkomba.xhtml\\ \hline
copper\_1922\_nu10 & 272 & 445 & https://www.kinsources.net/kidarep/dataset-36-copper-1922-nu10.xhtml\\ \hline
sarmi & 277 & 516 & https://www.kinsources.net/kidarep/dataset-213-sarmi.xhtml\\ \hline
dogrib\_1911\_25\_59\_nd04 & 278 & 466 & https://www.kinsources.net/kidarep/dataset-62-dogrib-1911-25-59-nd04.xhtml\\ \hline
konkama\_1951\_eu01 & 289 & 477 & https://www.kinsources.net/kidarep/dataset-9-konkama-1951-eu01.xhtml\\ \hline
tikopia\_1930 & 294 & 441 & https://www.kinsources.net/kidarep/dataset-18-tikopia-1930.xhtml\\ \hline
tory & 299 & 532 & https://www.kinsources.net/kidarep/dataset-13-tory.xhtml\\ \hline
nunamiut-tareumiut\_1900\_nu12 & 304 & 472 & https://www.kinsources.net/kidarep/dataset-42-nunamiut-tareumiut-1900-nu12.xhtml\\ \hline
mowanjum-kalumburu & 310 & 328 & https://www.kinsources.net/kidarep/dataset-209-mowanjum-kalumburu.xhtml\\ \hline
kung\_1952\_af01 & 313 & 552 & https://www.kinsources.net/kidarep/dataset-67-kung-1952-af01.xhtml\\ \hline
parakana & 330 & 623 & https://www.kinsources.net/kidarep/dataset-73-parakana.xhtml\\ \hline
wilcania & 337 & 572 & https://www.kinsources.net/kidarep/dataset-51-wilcania.xhtml\\ \hline
apache\_1936\_nd03 & 378 & 610 & https://www.kinsources.net/kidarep/dataset-52-apache-1936-nd03.xhtml\\ \hline
dogon-konsogu-donyu & 399 & 592 & https://www.kinsources.net/kidarep/dataset-204-dogon-konsogu-donyu.xhtml\\ \hline
kodiak & 410 & 746 & https://www.kinsources.net/kidarep/dataset-240-kodiak.xhtml\\ \hline
melombo & 435 & 672 & https://www.kinsources.net/kidarep/dataset-64-melombo.xhtml\\ \hline
makuna\_1973 & 436 & 761 & https://www.kinsources.net/kidarep/dataset-72-makuna-1973.xhtml\\ \hline
genesis & 439 & 628 & https://www.kinsources.net/kidarep/dataset-70-genesis.xhtml\\ \hline
arawete & 454 & 980 & https://www.kinsources.net/kidarep/dataset-74-arawete.xhtml\\ \hline
ojibwa\_1949\_nd08 & 479 & 831 & https://www.kinsources.net/kidarep/dataset-19-ojibwa-1949-nd08.xhtml\\ \hline
netsilik\_1922\_nu09 & 502 & 790 & https://www.kinsources.net/kidarep/dataset-34-netsilik-1922-nu09.xhtml\\ \hline
torres\_strait & 585 & 1251 & https://www.kinsources.net/kidarep/dataset-44-torres-strait.xhtml\\ \hline
nunivak & 619 & 1237 & https://www.kinsources.net/kidarep/dataset-251-nunivak.xhtml\\ \hline
chenchu\_1940\_as02 & 636 & 1154 & https://www.kinsources.net/kidarep/dataset-92-chenchu-1940-as02.xhtml\\ \hline
igluligmiut\_1961\_nu07 & 645 & 1099 & https://www.kinsources.net/kidarep/dataset-65-igluligmiut-1961-nu07.xhtml\\ \hline
nyungar & 657 & 1166 & https://www.kinsources.net/kidarep/dataset-27-nyungar.xhtml\\ \hline
anuta\_1972 & 659 & 1288 & https://www.kinsources.net/kidarep/dataset-3-anuta-1972.xhtml\\ \hline
mbuti\_forest\_1957\_af02 & 706 & 1184 & https://www.kinsources.net/kidarep/dataset-60-mbuti-forest-1957-af02.xhtml\\ \hline
yaraldi & 738 & 1212 & https://www.kinsources.net/kidarep/dataset-32-yaraldi.xhtml\\ \hline
trio\_1960s & 782 & 1368 & https://www.kinsources.net/kidarep/dataset-28-trio-1960s.xhtml\\ \hline
achuar\_pastaza & 795 & 1389 & https://www.kinsources.net/kidarep/dataset-150-achuar-pastaza.xhtml\\ \hline
nucoorilma\_tingha & 798 & 1416 & https://www.kinsources.net/kidarep/dataset-229-nucoorilma-tingha.xhtml\\ \hline
pakaa\_nova & 815 & 1603 & https://www.kinsources.net/kidarep/dataset-7-pakaa-nova.xhtml\\ \hline
la\_hague & 828 & 1363 & https://www.kinsources.net/kidarep/dataset-132-la-hague.xhtml\\ \hline
saudi\_royal\_genealogy & 868 & 981 & https://www.kinsources.net/kidarep/dataset-20-saudi-royal-genealogy.xhtml\\ \hline
surui & 926 & 1982 & https://www.kinsources.net/kidarep/dataset-68-surui.xhtml\\ \hline
achuar\_huasaga\_chankuap & 1140 & 2013 & https://www.kinsources.net/kidarep/dataset-173-achuar-huasaga-chankuap.xhtml\\ \hline
samburu & 1263 & 2021 & https://www.kinsources.net/kidarep/dataset-223-samburu.xhtml\\ \hline
ayd\_nl\_yoruk\_2005 & 1269 & 2395 & https://www.kinsources.net/kidarep/dataset-24-ayd-nl-yoruk-2005.xhtml\\ \hline
todas & 1423 & 3211 & https://www.kinsources.net/kidarep/dataset-258-todas.xhtml\\ \hline
alyawarra\_1818-1979\_au10 & 1461 & 2678 & https://www.kinsources.net/kidarep/dataset-88-alyawarra-1818-1979-au10.xhtml\\ \hline
nobles\_ile-de-france\_1000-1440 & 1463 & 1969 & https://www.kinsources.net/kidarep/dataset-306-nobles-ile-de-france-1000-1440.xhtml\\ \hline
mebengokre\_kayapo & 1466 & 2572 & https://www.kinsources.net/kidarep/dataset-320-mebengokre-kayapo.xhtml\\ \hline
omaha\_1880 & 1513 & 2217 & https://www.kinsources.net/kidarep/dataset-90-omaha-1880.xhtml\\ \hline
tikuna-arara & 1695 & 3206 & https://www.kinsources.net/kidarep/dataset-103-tikuna-arara.xhtml\\ \hline
candoshi & 1840 & 3649 & https://www.kinsources.net/kidarep/dataset-95-candoshi.xhtml\\ \hline
chuukese\_1947\_1940 & 2049 & 4159 & https://www.kinsources.net/kidarep/dataset-35-chuukese-1947-1940.xhtml\\ \hline
us-presidents & 2477 & 4015 & https://www.kinsources.net/kidarep/dataset-56-us-presidents.xhtml\\ \hline
kelkummer & 2588 & 5792 & https://www.kinsources.net/kidarep/dataset-61-kelkummer.xhtml\\ \hline
familles\_bretonnes\_xve-xviie\_siecles & 2680 & 4131 & https://www.kinsources.net/kidarep/dataset-177-familles-bretonnes-xve-xviie-siecles.xhtml\\ \hline
manus\_1929 & 2821 & 5122 & https://www.kinsources.net/kidarep/dataset-30-manus-1929.xhtml\\ \hline
bassari\_guinea & 2880 & 5262 & https://www.kinsources.net/kidarep/dataset-29-bassari-guinea.xhtml\\ \hline
warao & 2899 & 5847 & https://www.kinsources.net/kidarep/dataset-26-warao.xhtml\\ \hline
cocama-cocamilla & 2975 & 5108 & https://www.kinsources.net/kidarep/dataset-159-cocama-cocamilla.xhtml\\ \hline
torshan & 3008 & 6074 & https://www.kinsources.net/kidarep/dataset-80-torshan.xhtml\\ \hline
ammonni & 3014 & 5454 & https://www.kinsources.net/kidarep/dataset-128-ammonni.xhtml\\ \hline
feistritz\_am\_gael\_1990 & 3151 & 4289 & https://www.kinsources.net/kidarep/dataset-54-feistritz-am-gael-1990.xhtml\\ \hline
duu\_rea & 4109 & 6544 & https://www.kinsources.net/kidarep/dataset-287-duu-rea.xhtml\\ \hline
obidos & 4178 & 7354 & https://www.kinsources.net/kidarep/dataset-45-obidos.xhtml\\ \hline
charlevoix & 4463 & 8439 & https://www.kinsources.net/kidarep/dataset-115-charlevoix.xhtml\\ \hline
baruya & 5016 & 10 738 & https://www.kinsources.net/kidarep/dataset-249-baruya.xhtml\\ \hline
ancien\_regime & 5891 & 10 683 & https://www.kinsources.net/kidarep/dataset-53-ancien-regime.xhtml\\ \hline
ragusa & 5999 & 11 317 & https://www.kinsources.net/kidarep/dataset-47-ragusa.xhtml\\ \hline
samo & 6389 & 10 885 & https://www.kinsources.net/kidarep/dataset-117.xhtml\\ \hline
ebrei & 7331 & 11 959 & https://www.kinsources.net/kidarep/dataset-94-ebrei.xhtml\\ \hline
kel\_owey & 8809 & 15 645 & https://www.kinsources.net/kidarep/dataset-194-kel-owey.xhtml\\ \hline
sainte\_catherine & 9595 & 14 988 & https://www.kinsources.net/kidarep/dataset-93-sainte-catherine.xhtml\\ \hline
dogon\_boni & 11 294 & 25 430 & https://www.kinsources.net/kidarep/dataset-186-dogon-boni.xhtml\\ \hline
dynasties\_juives & 11 521 & 19 157 & https://www.kinsources.net/kidarep/dataset-125-dynasties-juives.xhtml\\ \hline
bwa\_slam\_biogsurvey & 18 645 & 32 726 & https://www.kinsources.net/kidarep/dataset-307-bwa-slam-biogsurvey.xhtml\\ \hline
san\_marino & 28 586 & 51 446 & https://www.kinsources.net/kidarep/dataset-76-san-marino.xhtml\\ \hline
watchi\_2017 & 50 519 & 83 376 & https://www.kinsources.net/kidarep/dataset-275-watchi-2017.xhtml\\ \hline

\end{longtable}

\bibliographystyle{elsarticle-num}
\bibliography{references}

\end{document}